\documentclass[hidelinks,11pt]{article}

\usepackage{amsmath}
\usepackage{amssymb}
\usepackage{amsfonts,amsthm}
\usepackage{thmtools,thm-restate}
\usepackage{tabularx,lipsum,environ}
\usepackage{multirow}
\usepackage{hyperref}

\newtheorem{theorem}{Theorem}
\newtheorem{lemma}{Lemma}
\newtheorem{Rule}{Rule}[theorem]
\newtheorem{Rulel}{Rule}[lemma]
\newtheorem{claim}{Claim}[theorem]
\newtheorem{corollary}{Corollary}

\newtheorem{observation}{Observation}
\newtheorem{proposition}{Proposition}

\newcommand{\Oh}{O}
\newcommand{\W}[1]{{W[#1]}}
\newcommand{\cW}[1]{{co-W[#1]}}

\newenvironment{claimproof}[1]{\par\noindent\emph{Proof:}\space#1}{{\leavevmode\unskip\penalty9999 \hbox{}\nobreak\hfill\quad\hbox{$\lrcorner$}}\medskip}

\usepackage{tikz}
\usepackage[a4paper,bindingoffset=0.2in,
            top=1.1in, bottom=1.1in, left=1.0in, right=1.0in,
            footskip=.5in]{geometry}

\newcommand{\ignore}[1]{}

\makeatletter
\newcommand{\problemtitle}[1]{\gdef\@problemtitle{#1}}% Store problem title
\newcommand{\probleminput}[1]{\gdef\@probleminput{#1}}% Store problem input
\newcommand{\problemquestion}[1]{\gdef\@problemquestion{#1}}% Store problem question
\NewEnviron{problem}{
  \problemtitle{}\probleminput{}\problemquestion{}% Default input is empty
  \BODY% Parse input
  \par\addvspace{.5\baselineskip}
  \noindent
  \begin{tabularx}{\textwidth}{@{\hspace{\parindent}} l X c}
    \multicolumn{2}{@{\hspace{\parindent}}l}{\@problemtitle} \\% Title
    \textbf{Input:} & \@probleminput \\% Input
    \textbf{Task:} & \@problemquestion% Question
  \end{tabularx}
  \par\addvspace{.5\baselineskip}
}

\title{Parameterized Aspects of Strong Subgraph Closure\thanks{This work is supported by Research Council of Norway via project ``CLASSIS''.}}

\author{
Petr A. Golovach\footnote{Department of Informatics, University of Bergen, Norway. Emails: \texttt{\{petr.golovach, pinar.heggernes, paloma.lima\}@uib.no}}
\and
Pinar Heggernes$^\dagger$
\and
Athanasios L. Konstantinidis\footnote{Department of Mathematics, University of Ioannina, Greece. Emails: \texttt{skonstan@cc.uoi.gr, charis@cs.uoi.gr}}
\and
Paloma T. Lima$^\dagger$
\and
Charis Papadopoulos$^\ddagger$
}

\date{}

\begin{document}
\maketitle

\begin{abstract}
Motivated by the role of triadic closures in social networks, and the importance of finding a maximum subgraph avoiding a fixed pattern, we introduce and initiate the parameterized study of the {\sc Strong $F$-closure} problem, where $F$ is a fixed graph. This is a generalization of {\sc Strong Triadic Closure}, whereas it is a relaxation of {\sc $F$-free Edge Deletion}. In {\sc Strong $F$-closure}, we want to select a maximum number of edges of the input graph $G$, and mark them as {\it strong edges}, in the following way: whenever a subset of the strong edges forms a subgraph isomorphic to $F$, then the corresponding induced subgraph of $G$ is {\it not} isomorphic to $F$. Hence the subgraph of $G$ defined by the strong edges is not necessarily $F$-free, but whenever it contains a copy of $F$, there are additional edges in $G$ to destroy that strong copy of $F$ in $G$.

We study \textsc{Strong $F$-closure} from a parameterized perspective with various natural parameterizations. Our main focus is on the number $k$ of strong edges as the parameter.
We show that the problem is FPT with this parameterization for every fixed graph $F$, whereas it does not admit a polynomial kernel even when $F =P_3$. In fact, this latter case is equivalent to the {\sc Strong Triadic Closure} problem, which motivates us to study this problem on input graphs belonging to well known graph classes. We show that {\sc Strong Triadic Closure} does not admit a polynomial kernel even when the input graph is a split graph, whereas it admits a polynomial kernel when the input graph is planar, and even $d$-degenerate. Furthermore, on graphs of maximum degree at most 4,  we show that {\sc Strong Triadic Closure} is FPT with the above guarantee parameterization $k - \mu(G)$, where $\mu(G)$ is the maximum matching size of $G$. We conclude with some results on the parameterization of \textsc{Strong $F$-closure} by the number of edges of $G$ that are not selected as strong.
\end{abstract}

\section{Introduction}
Graph modification problems are at the heart of parameterized algorithms. In particular, the problem of deleting as few edges as possible from a graph so that the remaining graph satisfies a given property has been studied extensively from the viewpoint of both classical and parameterized complexity for the last four decades \cite{CyganFKLMPPS15,DowneyF13,Yannakakis1981}. For a fixed graph $F$, a graph $G$ is said to be {\it F-free} if $G$ has no induced subgraph isomorphic to $F$. The \textsc{$F$-Free Edge Deletion} problem asks for the removal of minimum number of edges from an input graph $G$ so that the remaining graph is $F$-free. In this paper, we introduce a relaxation of this problem, which we call \textsc{Strong $F$-closure}. Our problem is also a generalization of the {\sc Strong Triadic Closure} problem, which asks to select as many edges as possible of a graph as {\it strong}, so that whenever two strong edges $uv$ and $vw$ share a common
endpoint $v$, the edge $uw$ is also present in the input graph (not necessarily strong). This problem is well studied in the area of social networks \cite{EK10,BK14}, and its classical computational complexity has been studied recently both on general graphs and on particular graph classes \cite{ST14,KP17}.

In the \textsc{Strong $F$-closure} problem, we have a fixed graph $F$, and we are given an input graph $G$, together with an integer $k$. The task is to decide whether we can select at least $k$ edges of $G$ and mark them as {\it strong}, in the following way: whenever the subgraph of $G$ spanned by the strong edges contains an induced subgraph isomorphic to $F$, then the corresponding induced subgraph of $G$ on the same vertex subset is not isomorphic to $F$. The remaining edges of $G$ that are not selected as strong, will be called {\it weak}. Consequently, whenever a subset $S$ of the strong edges form a copy of $F$, there must be an additional strong or weak edge in $G$ with endpoints among the endpoints of edges in $S$. A formal definition of the problem is easier to give via spanning subgraphs. If two graphs $H$ and $F$ are isomorphic then we write $H \simeq F$, and if they are not isomorphic then we write $H \not\simeq F$. Given a graph $G$ and a fixed graph $F$, we say that a (not necessarily induced)
 subgraph $H$ of $G$ \emph{satisfies the $F$-closure} if, for every $S \subseteq V(H)$ with $H[S] \simeq F$, we have that $G[S] \not\simeq F$. In this case, the edges of $H$ form exactly the set of strong edges of $G$.

\begin{problem}
 \problemtitle{\textsc{Strong $F$-closure}}
  \probleminput{A graph $G$ and a nonnegative integer $k$.}
  \problemquestion{Decide whether $G$ has a spanning subgraph $H$ that satisfies the $F$-closure, such that $|E(H)| \ge k$.}
\end{problem}

Based on this definition and the above explanation, the terms ``marking an edge as weak (in $G$)'' and ``removing an edge (of $G$ to obtain $H$)'' are equivalent, and we will use them interchangeably. An induced path on three vertices is denoted by $P_3$. Relating \textsc{Strong $F$-closure} to the already mentioned problems, observe that \textsc{Strong $P_3$-closure}  is exactly {\sc Strong Triadic Closure}. Observe also that
a solution for \textsc{$F$-free Edge Deletion} is a solution for  \textsc{Strong $F$-closure}, since the removed edges in the first problem can simply be taken as the weak edges in the second problem. However it is important to note that the reverse is not always true. All of the mentioned problems are known to be NP-hard. The paremeterized complexity of \textsc{$F$-free Edge Deletion} has been studied extensively when parameterized by $\ell =$ the number of removed edges. With this parameter, the problem is FPT if $F$ is of constant size \cite{Cai96}, whereas it becomes W[2]-hard when parameterized by the size of $F$ even for $\ell=0$~\cite{KhotR02}.
 Moreover, there exists a small graph $F$ on seven vertices for which \textsc{$F$-free Edge Deletion} does not admit a polynomial kernel \cite{KW13} when the problem is parameterized by $\ell$.
To our knowledge, {\sc Strong Triadic Closure} has not been studied with respect to parameterized complexity before our work.

In this paper, we study the parameterized complexity of {\sc Strong $F$-closure} with three different natural parameters:
the number of strong edges,
the number of strong edges above guarantee (maximum matching size),
and the number of weak edges.
\begin{itemize}
\item
In Section~\ref{sec:strong}, we show that  {\sc Strong $F$-closure} is FPT when parameterized by $k=|E(H)|$ for a fixed $F$. Moreover, we prove that the problem is FPT even when we allow the size of $F$ to be a parameter, that is, if we parameterize the problem by $k+|V(F)|$, except if $F$ has at most one edge. In the latter case {\sc Strong $F$-closure} is \W{1}-hard when parameterized by $|V(F)|$ even if $k\leq 1$. We also observe that {\sc Strong $F$-closure} parameterized by $k+|V(F)|$ admits a polynomial kernel if $F$ has a component with at least three vertices and the input graph is restricted to be $d$-degenerate.
This result is tight in the sense that it cannot be generalized to nowhere dense graphs.
\item In Section~\ref{sec:stc}, we focus on the case $F=P_3$, that is, we investigate the parameterized complexity of {\sc Strong Triadic Closure}. We complement the FPT results of the previous section by proving that  {\sc Strong Triadic Closure} does not admit a polynomial kernel even on split graphs.
It is straightforward to see that if $F$ has a connected component on at least three vertices, then a matching in $G$ gives a feasible solution for {\sc Strong $F$-closure}.
Thus the maximum matching size $\mu(G)$ provides a lower bound for the maximum number of edges of $H$.
Consequently, parameterization above this lower bound becomes interesting. 
Motivated by this, we study {\sc Strong $F$-closure} parameterized by $|E(H)|-\mu(G)$.
It is known that
{\sc Strong Triadic Closure} can be solved in polynomial time on subcubic graphs, but it is NP-complete on graphs of maximum degree at most $d$ for every $d\geq 4$ \cite{KonstantinidisN17}.
As a first step in the investigation of the parameterization above lower bound, we show that  {\sc Strong Triadic Closure} is FPT  on graphs of maximum degree at most 4, parameterized by $|E(H)|-\mu(G)$. 
\item
 Finally, in Section~\ref{sec:conc}, we consider {\sc Strong $F$-closure} parameterized by $\ell=|E(G)|-|E(H)|$, that is, by the number of weak edges. We show that the problem is FPT and admits a polynomial (bi-)kernel if $F$ is a fixed graph. Notice that, contrary to the parameterization by $k+|V(F)|$, we cannot hope for FPT results when the problem is parameterized by $\ell+|V(F)|$. This is because, when $\ell=0$, \textsc{Strong $F$-closure} is equivalent to asking whether $G$ is $F$-free, which is equivalent to solving \textsc{Induced Subgraph Isomorphism} that is well known to be \W{1}-hard~\cite{DowneyF13,KhotR02}. We also state some additional results and open problems.
 Our findings are summarized in Table \ref{tab:results}.
\end{itemize}

\begin{table}
\begin{center}
\setlength\extrarowheight{3pt}
\begin{tabular}{|l l l c|}\hline
 Parameter & Restriction & Parameterized Complexity & Theorem \\\hline
\multirow{4}{*}{$|E(H)|+|V(F)|$}  & $|E(F)|\leq 1$   & \W{1}-hard        & \ref{prop:empty}, \ref{prop:one-edge}\\\cline{2-4}
                                             &  $|E(F)|\geq 2$   & FPT        & \ref{thm:FPT-strong}\\\cline{2-4}

                                             & $F$ has a cc with 3 vertices,                     & \multirow{2}{*}{polynomial kernel} & \multirow{2}{*}{\ref{prop:bigcomp-kern}}\\[-3pt]
                                             & $G$ is $d$-degenerate             &                & \\\cline{1-4}

 \multirow{2}{*}{$|E(H)|$}                   & $F$ has no isolated vertices      & FPT   & \ref{cor:no-isolates}\\\cline{2-4}
                                             & $F=P_3$, $G$ is split             & {no polynomial kernel} & {\ref{theo:strong:nopolykernel:P3:split}}\\\cline{1-4}

 {$|E(H)|-\mu(G)$} & $F=P_3$, $\Delta(G)\leq 4$    & {FPT} & {\ref{theo:strong:bounded4:fpt}}\\\cline{1-4}

 \multirow{2}{*}{$|E(G)|-|E(H)|$   }
					      & \multirow{2}{*}{None} & FPT          & \ref{theo:weak:bounded:fpt}\\\cline{3-4}
                                             &                              & polynomial (bi-)kernel & \ref{theo:weak:bounded:polykernel}\\
\hline

                                             \end{tabular}
\end{center}
\vspace*{-0.05in}
\caption{Summary of our results: parameterized complexity analysis of {\sc Strong $F$-closure}.}
\label{tab:results}
\vspace*{-0.1in}
\end{table}

%%%%%%%%%%%%%%%%%%%%%%%%%%%%%%%%%%%%%%%%%%%%%%%%%%%%%%%%%%%%%%%%%%%%%%%%%%%%%%%%%%%%%%%%%%%%
\vspace*{-0.1in}
\section{Preliminaries}
All graphs considered here are simple and undirected.
We refer to Diestel's classical book \cite{Diestel12} for standard graph terminology that is undefined here.
Given an input graph $G$, we use the convention that $n=|V|$ and $m=|E|$.
For a graph $F$, it is said that $G$ is \emph{$F$-free} if $G$ has no induced subgraph isomorphic to $F$.
For a positive integer $d$,  $G$ is \emph{$d$-degenerate} if every subgraph of $G$ has a vertex of degree at most $d$. The maximum degree of  $G$ is denoted by $\Delta(G)$.
We denote by $G+H$ the disjoint union of two graphs $G$ and $H$. For a positive integer $p$, $pG$ denotes the disjoint union of $p$ copies of $G$.
A \emph{matching} in $G$ is a set of edges having no common endpoint.
The \emph{maximum matching number}, denoted by $\mu(G)$, is the maximum number of edges in any matching of $G$.
We say that a vertex $v$ is {\it covered} by a matching $M$ if $v$ is incident to an edge of $M$.
An \emph{induced matching}, denoted by $qK_2$, is a matching $M$ of $q$ edges such that $G[V(M)]$ is isomorphic to $qK_2$.

Let us give a couple of observations on the nature of our problem. An $F$-graph of a subgraph $H$ of $G$ is an induced subgraph $H[S] \simeq F$ such that $G[S] \simeq F$.
Clearly, if $H$ is a solution for \textsc{Strong $F$-closure} on $G$, then there is no $F$-graph in $H$, even though $H$ might have induced subgraphs isomorphic to $F$.
For \textsc{$F$-free Edge Deletion,} note that the removal of an edge that belongs to a forbidden subgraph might generate a new forbidden subgraph.
However, for \textsc{Strong $F$-closure} problem, it is not difficult to see that the removal of an edge that belongs to an $F$-graph cannot create a new critical subgraph.

\begin{observation}\label{obs:subgraph}
Let $G$ be a graph, and  let $H$ and $H'$ be spanning subgraphs of $G$ such that $E(H')\subseteq E(H)$. If $H$ satisfies the $F$-closure for some $F$, then $H'$ satisfies the $F$-closure.
\end{observation}

In particular, Observation~\ref{obs:subgraph} immediately implies that if an instance of \textsc{Strong $F$-closure} has a solution, it has a solution with \emph{exactly} $k$ edges.

We conclude this section with some definitions from parameterized complexity and kernelization. A problem with input size $n$ and parameter $k$ is \emph{fixed parameter tractable} (FPT), if it can be solved in time $f(k) \cdot n^{O(1)}$ for some computable function $f$.
A \emph{bi-kernelization}~\cite{AlonGKSY11} (or \emph{generalized kernelization}~\cite{BodlaenderDFH09}) for a parameterized problem $P$ is a polynomial algorithm that maps each instance $(x, k)$ of $P$ with the input $x$ and the parameter $k$ into
to an instance $(x', k')$ of some parameterized problem $Q$ such that
i) $(x, k)$ is a yes-instance of $P$ if and only if $(x', k')$ is a yes-instance of $Q$,
ii) the size of $x'$ is bounded by $f(k)$ for a computable function $f$, and
iii) $k'$ is bounded by $g(k)$ for a computable function $g$.
The output $(x', k')$ is called a \emph{bi-kernel} (or \emph{generalized kernel}) of the considered problem.
The function $f$ defines the size of a bi-kernel and the \emph{bi-kernel has polynomial size} if the function $f$ is polynomial.
If $Q=P$, then bi-kernel is called \emph{kernel}. Note that if $Q$ is in NP and $P$ is NP-complete, then the existence of a polynomial bi-kernel implies that $P$ has a polynomial kernel because there exists a polynomial reduction of $Q$ to $P$.
A \emph{polynomial compression} of a parameterized problem $P$ into a (nonparameterized) problem $Q$ is a polynomial algorithm that takes as an input an instance $(x,k)$ of $P$ and returns an instance $x'$ of $Q$ such that i) $(x, k)$ is a yes-instance of $P$ if and only if $x'$ is a yes-instance of $Q$,
ii) the size of $x'$ is bounded by $p(k)$ for a polynomial $p$.
For further details on parameterized complexity we refer to \cite{CyganFKLMPPS15,DowneyF13}.

%%%%%%%%%%%%%%%%%%%%%%%%%%%%%%%%%%%%%%%%%%%%%%%%%%%%%%%%%%%%%%%%%%%%%%%%%%%%%%%%%%%%%%%%%%%%

%%%%%%%%%%%%%%%%%%%%%%%%%%%%%%%%%%%%%%%%%%%%%%%%%%%%%%%%%%%%%%%%%%%%%%%%%%%%%%%%%%%%%%%%%%%%

\section{Parameterized complexity of Strong F-closure}
\label{sec:strong}

In this section we give a series of lemmata, which together lead to the conclusion that {\sc Strong $F$-closure} is FPT when parameterized by $k$. Observe that in our definition of the problem, $F$ is a fixed graph of constant size. However, the results of this section allow us to also take the size of $F$ as a parameter, making the results more general. We start by making some observations that will rule out some simple types of graphs as $F$.

\begin{observation}\label{obs:empty}
Let $p$ be a positive integer.
A graph $G$ has a spanning subgraph $H$ satisfying the $pK_1$-closure if and only if $G$ is $pK_1$-free, and if $G$ is $pK_1$-free, then every spanning subgraph $H$ of $G$ satisfies the $pK_1$-closure.
\end{observation}

By combining Observation~\ref{obs:empty} and the well known result that \textsc{Independent Set} is \W{1}-hard
when parameterized by the size of the independent set~\cite{DowneyF13}, we obtain the following:

\begin{proposition}\label{prop:empty}
For a positive integer $p$, \textsc{Strong $pK_1$-closure} can be solved in time $n^{\Oh(p)}$, and it is \cW{1}-hard for $k\geq 0$ when parameterized by $p$.
\end{proposition}

Using Proposition~\ref{prop:empty}, we assume throughout the remaining parts of the paper that every considered graph $F$ has at least one edge. We have another special case $F=pK_1+K_2$.

\begin{proposition}\label{prop:one-edge}
For a nonnegative integer $p$, \textsc{Strong $(pK_1+K_2)$-closure} can be solved in time $n^{\Oh(p)}$, and it is \cW{1}-hard for $k\geq 1$ when parameterized by $p$.
\end{proposition}

\begin{proof}
Let $F=pK_1+K_2$. If $p=0$, then $(G,k)$ is a yes-instance of \textsc{Strong $F$-closure} if and only if $k=0$. Assume that $p\geq 1$. Let $H$ be a spanning subgraph of $G$. Notice that $H$ satisfies the $F$-closure if and only if for every edge $uv$ of $H$, $G-N[\{u,v\}]$ has no independent set of size $p$.

This observation implies that to find a spanning subgraph $H$ of $G$ satisfying the $F$-closure, we can use the following procedure: for every edge $uv\in E(G)$, we check whether $G-N[\{u,v\}]$ has an independent set with  $p$ vertices, and then if this holds, we discard $uv$, and we include $uv$ in the set of edges of $H$ otherwise.  Clearly, it can be done in time $n^{\Oh(p)}$.

To show hardness, we reduce \textsc{Independent Set}. For simplicity, we prove the claim for $k=1$.
Let $(G,p)$ an instance of \textsc{Independent Set}. Let $Q$ be the graph obtained from two copies of the star $K_{1,p}$ by making their central vertices $u$ and $v$ adjacent. We define $G'=G+Q$. We claim that $G'$ has a spanning subgraph $H$ satisfying the $F$-closure that has exactly one edge if and only if $G$ has no independent set with $p$ vertices.
Suppose that $G$ has no independent set with $p$ vertices. Then the spanning subgraph $H$ of $G$ with $E(H)=\{uv\}$ satisfies the $F$-closure. Assume now that $H$ is a spanning subgraph of $G$ with $E(H)=\{xy\}$. Suppose that $xy\neq uv$. Assume without loss of generality that $u$ is not an end-vertex of $xy$. Then $u$ has $p$ neighbors distinct from $v$ that form an independent set in $G'-N[\{x,y\}]$ but this contradicts the property that $H$ satisfies the $F$-closure. Hence, $xy=uv$. Then $G=G'-N[\{u,v\}]$ has no independent set with $p$ vertices.
By Observation~\ref{obs:subgraph}, we have that $(G,p)$ is a no-instance of \textsc{Independent Set} if and only if $(G',k)$ is a yes-instance of \textsc{Strong $F$-closure}.
\end{proof}

From now on we assume that $F\neq pK_1$ and $F \neq pK_1+K_2$. We show that \textsc{Strong $F$-closure} is FPT when parameterized by $k$ and $|V(F)|$ in this case. We will consider separately the case when $F$ has a connected component with at least 3 vertices and the case $F=pK_1+qK_2$ for $p\geq 0$ and $q\geq 2$.

\begin{lemma}\label{lem:bigcomp}
Let $F$ be a graph that has a connected component with at least 3 vertices. Then  \textsc{Strong $F$-closure} can be solved in time $2^{\Oh(k^2)}(|V(F)|+k)^{\Oh(k)} + n^{\Oh(1)}$.
\end{lemma}

\begin{proof}
We show the claim by proving that the problem has a kernel with at most $2^{2k-2}(|V(F)|+k)+2k-2$ vertices.
Let $(G,k)$ be an instance of \textsc{Strong $F$-closure}. We recursively apply the following reduction rule in $G$:
\begin{Rulel}\label{rule:twins}
If there are at least $|V(F)|+k+1$ false twins in $G$, then  remove one of them.
\end{Rulel}

To show that the rule is sound, let $v_1,\ldots,v_{p}$ be false twins of $G$ for $p=\max\{|V(F)|,k\}+1$ and assume that $G'$ is obtained from $G$ by deleting $v_p$. We claim that $(G,k)$ is a yes-instance of \textsc{Strong $F$-closure} if and only if $(G',k)$ is a yes-instance.

Let $(G,k)$ be a yes-instance. By Observation~\ref{obs:subgraph},  there is a solution $H$ for $(G,k)$ such that $|E(H)|=k$. Since $|E(H)|=k$, there is $i\in\{1,\ldots,p\}$ such that $v_i$ is an isolated vertex of $H$. Since $v_1,\ldots,v_p$ are false twins we can assume without loss of generality that  $i=p$. Then $H'=H-v_p$ is a solution for $(G',k)$, that is, this is a yes-instance.
Assume that $(G',k)$ is a yes-instance of \textsc{Strong $F$-closure}. Let $H'$ be a solution for the instance with $k$ edges.
Denote by $H$ the spanning subgraph of $G$ with $E(H)=E(H')$. We show that $H$ satisfies the $F$-closure with respect to $G$. To obtain a contradiction, assume that there is a set of vertices $S$ of $G$ such that $H[S]\simeq F$ and $G[S]\simeq F$. Since $H'$ satisfies the $F$-closure with respect to $G$, $v_p\in S$. Note that $v_p$ is an isolated vertex of $H$. Because $p=|V(F)|+k+1$, there is $i\in \{1,\ldots,p-1\}$ such that $v_i$ is an isolated vertex of $H$ and $v_i\notin S$. Let $S'=(S\setminus\{v_p\})\cup \{v_i\}$. Since $v_i$ and $v_p$ are false twins,
$H[S']=H'[S']\simeq F$ and $G[S']\simeq F$; a contradiction. Therefore, we conclude that  $H$ satisfies the $F$-closure with respect to $G$, that is, $H$ is a solution for $(G,k)$.

It is straightforward to see that the rule can be applied in polynomial time.
To simplify notations, assume that $(G,k)$ is the instance of \textsc{Strong $F$-closure} obtained by the exhaustive application of Rule~\ref{rule:twins}.
We greedily find an inclusion maximal matching $M$ in $G$. Notice that the spanning subgraph $H$ of $G$ with $E(H)=M$ satisfies the $F$-closure because every component of $H$ has at most two vertices and by the assumption of the lemma $F$ has a component with at least 3 vertices. Therefore, if $|M|\geq k$, we have that $H$ is a solution for the instance. Respectively, we return $H$ and stop.

Assume that $|M|\leq k-1$. Let $X$ be the set of end-vertices of the edges of $M$. Clearly, $|X|\leq 2k-2$ and $X$ is a vertex cover of $G$.
Let $Y=V(G)\setminus X$. We have that $Y$ is an independent set. Every vertex in $Y$ has its neighbors in $X$. Hence, there are at most $2^{|X|}$ vertices of $Y$ with pairwise distinct neighborhoods. Hence, the vertices of $Y$ can be partitioned into at most $2^{|X|}$ classes of false twins. After applying Rule~\ref{rule:twins}, each class of false twins has at most $|V(F)|+k$ vertices. It follows that $|Y|\leq 2^{|X|}(|F(V)|+k)$ and
$$|V(G)|=|X|+|Y|\leq |X|+2^{|X|}(|F(V)|+k)\leq (2k-2)+2^{2k-2}(|F(V)|+k).$$

Now we can find a solution for $(G,k)$ by brute force checking all subsets of edges of size $k$ by Observation~\ref{obs:subgraph}. This can be done it time $|V(G)|^{\Oh(k)}$. Hence, the total running time is  $2^{\Oh(k^2)}(|V(F)|+k)^{\Oh(k)} + n^{\Oh(1)}$.
\end{proof}

Now we consider the case $F=pK_1+qK_2$ for $p\geq 0$ and $q\geq 2$. First, we explain how to solve  \textsc{Strong $qK_2$-closure} for $q\geq 2$.
We use the random separation technique proposed by Cai, Chen and Chan~\cite{CCC06} (see also \cite{CyganFKLMPPS15}).
To avoid dealing with randomized algorithms and subsequent standard derandomization we use the following lemma stated in~\cite{ChitnisCHPP16}. %by Chitnis et al.

\begin{lemma}[\cite{ChitnisCHPP16}]\label{lem:derand}
Given a set $U$ of size $n$ and integers $0\leq a,b\leq n$, one can construct in time $2^{\Oh(\min\{a,b\}\log (a+b))}\cdot n\log n$ a family $\mathcal{S}$ of at most  $2^{\Oh(\min\{a,b\}\log (a+b))}\cdot \log n$ subsets of $U$ such that the following holds: for any sets $A,B\subseteq U$, $A\cap B=\emptyset$, $|A|\leq a$, $|B|\leq b$, there exists a set $S\in \mathcal{S}$ with $A\subseteq S$ and $B\cap S=\emptyset$.
\end{lemma}

\begin{lemma}\label{lem:qK2}
For $q\geq 2$,
 \textsc{Strong $qK_2$-closure} can be solved in time $2^{\Oh(k\log k)}\cdot n^{\Oh(1)}$.
\end{lemma}

\begin{proof}
Let  $(G,k)$ be an instance of \textsc{Strong $qK_2$-closure}.
If $k<q$, then every spanning subgraph $H$ of $G$ with $k$ edges satisfies the $F$-closure, that is, $(G,k)$ is a yes-instance of  \textsc{Strong $F$-closure} if $k\leq |E(G)|$. Assume from now that $q\leq k$.

Suppose that $G$ has a vertex $v$ of degree at least $k$. Let $X$ be the set of edges of $G$ incident to $v$ and consider the spanning subgraph $H$ of $G$ with $E(H)=X$. Since $F=qK_2$  and $q\geq 2$, $H$ satisfies the $F$-closure. Hence, $H$ is a solution for $(G,k)$. We assume that this not the case and $\Delta(G) \le k-1$.

Suppose that $(G,k)$ is a yes-instance. Then by Observation~\ref{obs:subgraph}, there is a solution $H$ with exactly $k$ edges.
Let $A=E(H)$ and denote by $X$ the set of end-vertices of the edges of $A$. Denote by $B$ the set of edges of $E(G)\setminus A$ that have at least one end-vertex in $N[X]$. Clearly, $A\cap B=\emptyset$. We have that $|A|=k$ and because the maximum degree of $G$ is at most $k-1$, $|B|\leq 2k(k-1)(k-2)$. Applying Lemma~\ref{lem:derand} for the universe $U=E(G)$, $a=k$ and $b=2k(k-1)(k-2)$ , we construct in time  $2^{\Oh(k\log k)}\cdot n^{\Oh(1)}$ a family $\mathcal{S}$ of at most  $2^{\Oh(k\log k)}\cdot \log n$ subsets of $E(G)$ such that there exists a set $S\in \mathcal{S}$ with $A\subseteq S$ and $B\cap S=\emptyset$.
For every $S\in\mathcal{S}$, we find (if it exists) a spanning subgraph $H$ of $G$ with $k$ edges such that (i) $E(H)\subseteq S$ and (ii) for every $e_1,e_2\in S$ that are adjacent or have adjacent end-vertices, it holds that either $e_1,e_2\in E(H)$ or $e_1,e_2\notin E(H)$.
By Lemma~\ref{lem:derand}, we have that if  $(G,k)$ is a yes-instance of \textsc{Strong $F$-closure}, then it has a solution satisfying (i) and (ii).
Hence, if we find a solution for some $S\in\mathcal{S}$, we return it and stop and, otherwise, if there is no solution satisfying (i) and (ii) for some $S\in\mathcal{S}$, we conclude that $(G,k)$ is a no-instance.

Assume that $S\in\mathcal{S}$ is given. We describe the algorithm for finding a solution $H$ with $k$ edges satisfying (i) and (ii).
Let $R$ be the set of end-vertices of the edges of $S$. Consider the graph $G[R]$ and denote by $C_1,\ldots,C_r$ its components.
Let $A_i=E(C_i)\cap S$ for $i\in\{1,\ldots,r\}$.

Observe that if $H$ is a solution with $k$ edges satisfying (i) and (ii), then for each $i\in\{1,\ldots,r\}$, either $A_i\subseteq E(H)$ or $A_i\cap E(H)=\emptyset$. It means that we are looking for a solution $H$ such that $E(H)$ is union of some sets $A_i$, that is, $E(H)=\cup_{i\in I}A_i$ for $I\subseteq \{1,\ldots,r\}$.
Let $c_i=|A_i|$ for $i\in \{1,\ldots,r\}$. Clearly, we should have that $\sum_{i\in I}c_i=k$. In particular, it means that if $|A_i|>k$, then the edges of $A_i$ are not in any solution. Therefore, we discard such sets and assume from now that $|A_i|\leq k$ for $i\in\{1,\ldots,r\}$.
For $i\in\{1,\ldots,r\}$, denote by $w_i$ the maximum number of edges in $A_i$ that form an induced matching in $C_i$. Since each $|A_i|\leq k$, the values of $w_i$ can be computed in time $2^k\cdot n^{\Oh(1)}$ by brute force. Observe that for distinct $i,j\in \{1,\ldots,r\}$, the vertices of $C_i$ and $C_j$ are at distance at least two in $G$ and, therefore, the end-vertices of edges of $A_i$ and $A_j$ are not adjacent. It follows, that the problem of finding a solution $H$ is equivalent to the following problem: find $I\subseteq\{1,\ldots,r\}$ such that $\sum_{i\in I}c_i=k$ and $\sum_{i\in I}w_i\leq q$.
It is easy to see that we obtain an instance of a variant of the well known \textsc{Knapsack} problem (see, e.g.,~\cite{KleinbergT06}); the only difference is that we demand  $\sum_{i\in I}c_i=k$ instead of $\sum_{i\in I}c_i\geq k$ as in the standard version. This problem can be solved by the standard dynamic programming algorithm (again see, e.g.,~\cite{KleinbergT06}) in time $\Oh(kn)$.

Since the family $\mathcal{S}$ is constructed in time $2^{\Oh(k\log k)}\cdot n^{\Oh(1)}$ and we consider  $2^{\Oh(k\log k)}\cdot \log n$ sets $S$, we obtain that the total running time is
$2^{\Oh(k\log k)}\cdot n^{\Oh(1)}$.
\end{proof}

We use Lemma~\ref{lem:qK2} to solve \textsc{Strong $(pK_1+qK_2)$-closure}.

\begin{lemma}\label{lem:pK1qK2}
For $p\geq 0$ and $q\geq 2$,
 \textsc{Strong $(pK_1+qK_2)$-closure} can be solved in time $2^{\Oh((k+p)\log(k+p))}\cdot n^{\Oh(1)}$.
\end{lemma}

\begin{proof}
Let $F=pK_1+qK_2$. If $p=0$, we can apply Lemma~\ref{lem:qK2} directly. Assume that $p\geq 1$.
Let $(G,k)$ be an instance of \textsc{Strong $F$-closure}.
If $k<q$, then every spanning subgraph $H$ of $G$ with $k$ edges satisfies the $F$-closure, that is, $(G,k)$ is a yes-instance of  \textsc{Strong $F$-closure} if $k\leq |E(G)|$. Assume from now that $q\leq k$.

Suppose that $G$ has a vertex $v$ of degree at least $k$. Then we argue in exactly the same way as in the proof of Lemma~\ref{lem:qK2}.
We consider the set of edges $X$ incident to $v$ and define $H$ be the spanning subgraph of $G$ with $E(H)=X$.  Since $q\geq 2$, $H$ satisfies the $F$-closure and we have that  $H$ is a solution for $(G,k)$. We assume from now that this not the case and $\Delta(G) \le k-1$.

Suppose that $|V(G)|<2k(k-1)+pk$. In this case we solve \textsc{Strong $F$-closure} by brute force trying all possible subsets $X$ of $k$ edges and checking whether the spanning subgraph $H$ with $E(H)=X$ is a solution. By Observation~\ref{obs:subgraph}, it is sufficient to solve the problem.  To check whether $H$ is a solution, we have to verify whether $H$ satisfies the $F$-closure. We do it by brute force in time $n^{\Oh(|V(F)|)}$. Since $n\leq 2k(k-1)+pk$ and $|V(F)|=p+2q\leq p+2k$, this can be done in time $2^{\Oh((k+p)\log(k+p)}$. Since the number of sets $X$ is  $2^{\Oh((k+p)\log(k+p)}$, the total running time is $2^{\Oh((k+p)\log(k+p)}$.

Assume now that $|V(G)|\geq 2k(k-1)+pk$.

We claim that in this case a spanning subgraph $H$ of $G$ satisfies the $pK_1+qK_2$-closure if and only if
$H$ satisfies the $qK_2$-closure. It is straightforward to see that if $H$ satisfies the $qK_2$-closure, then $H$ satisfies the $pK_1+qK_2$-closure. Suppose that
$H$ does not satisfy the $qK_2$-closure. Then there is $S\subseteq V(G)$ of size $2q$ such that $G[S]=H[S]$ is a matching with $q$ edges. Let $X=V(G)\setminus N[S]$.
Since $\Delta(G) \le k-1$, $|N[S]|\leq 2k(k-1)$ and, therefore, $|X|\geq pk$. It implies that $G[X]$ has an independent set $S'$ of size at least $p$ because the maximum degree is bounded by $k-1$. We have that $G[S\cup S']=H[S\cup S']\simeq pK_1+qK_2$. It means that $H$ does not satisfy the $pK_1+qK_2$-closure.

By the proved claim, we have to solve  \textsc{Strong $qK_2$-closure} and this can be done in time  $2^{\Oh(k\log k)}\cdot n^{\Oh(1)}$ by Lemma~\ref{lem:qK2}.
\end{proof}

Combining Lemmata~\ref{lem:bigcomp}, \ref{lem:qK2}, and \ref{lem:pK1qK2}, we obtain the following theorem.

\begin{theorem}\label{thm:FPT-strong}
If $F\neq pK_1$ for $p\geq 1$ and $F\neq pK_1+K_2$ for $p\geq 0$, then \textsc{Strong $F$-closure} is FPT when parameterized by $|V(F)|+k$.
\end{theorem}

Notice that if $|E(F)|>k$, then $(G,k)$ is a yes-instance of \textsc{Strong $F$-closure}. This immediately implies the following corollary.

\begin{corollary}\label{cor:no-isolates}
If $F$ has no isolated vertices,  then \textsc{Strong $F$-closure} is FPT when parameterized by $k$,  even when $F$ is given as a part of the input.
\end{corollary}

We conclude this section with a kernel result. Observe that if the input graph $G$ is restricted to be a graph from a sparse graph class $\mathcal{C}$, namely if $\mathcal{C}$  is \emph{nowhere dense} (see \cite{NesetrilM12})
and is closed under taking subgraphs, then the kernel constructed in Lemma~\ref{lem:bigcomp} becomes polynomial. This observation is based on the results Eickmeyer et al.~\cite{EickmeyerGKKPRS17}  that allow to bound the number of distinct neighborhoods of vertices in $V(G)\setminus X$ in the construction of the kernel.  For simplicity, we demonstrate it here on $d$-degenerate graphs \footnote{NP-completeness result for $F=P_3$ restricted to planar graphs (and, thus, 5-degenerate graphs) is given in Section~\ref{sec:conc}.}.

\begin{proposition}\label{prop:bigcomp-kern}
If $F$ has a connected component with at least 3 vertices, then  \textsc{Strong $F$-closure} has a kernel with $k^{\Oh(d)}d(|V(F)|+k)$ vertices on $d$-degenerate graphs.
\end{proposition}

\begin{proof}
Let $(G,k)$ be an instance of \textsc{Strong $F$-closure} and $G$ is $d$-degenerate. First, we exhaustively apply Rule~\ref{rule:twins}. To simplify notations, assume that $(G,k)$ is the obtained instance. Then we  find an inclusion maximal matching $M$ in $G$. If $|M|\geq k$, we have that $H$ is a solution for the instance. Respectively, we return $H$ and stop.
Assume that this is not the case, that is, $|M|\leq k-1$. Let $X$ be the set of end-vertices of the edges of $M$. Clearly, $|X|\leq 2k-2$ and $X$ is a vertex cover of $G$.
Let $Y=V(G)\setminus X$. We have that $Y$ is an independent set.

Observe that if $Y$ contains at least $\binom{|X|}{d+1}d+1$ vertices of degree at least $d+1$, then $G$ contains the complete bipartite graph $K_{d+1,d+1}$ as a subgraph contradicting $d$-degeneracy. We conclude that $Y$ contains $d\cdot k^{\Oh(d)}$ vertices of degree at least $d+1$. The number of vertices of degree at most $d$ with pairwise distinct neighborhoods is $k^{\Oh(d)}$. This immediately implies that $G$ has $k^{\Oh(d)}d(|V(F)|+k)$ vertices.
\end{proof}

In particular, we have a polynomial kernel when $F=P_3$.
Similar results can be obtained for some classes of dense graphs. For example, if $G$ is $dK_1$-free, then $V(G)\setminus X$ has at most $d-1$ vertices and we obtain a kernel with $2k+d-3$ vertices.

\section{Parameterized complexity of \textsc{Strong Triadic Closure}}\label{sec:stc}
In this section we study the parameterized complexity of \textsc{Strong $P_3$-closure}, which is more famously known as \textsc{Strong Triadic Closure}.

Note that \textsc{Strong Triadic Closure} is FPT and admits an algorithm with running time $2^{k^2} \cdot n^{\mathcal{O}(1)}$ by Lemma~\ref{lem:bigcomp}. We complement this result by showing that \textsc{Strong Triadic Closure} does not admit a polynomial kernel, even when the input graph is a split graph. A graph is a {\it split graph} if its vertex set can be partitioned into an independent set and a clique. \textsc{Strong Triadic Closure} is known to be NP-hard on split graphs \cite{KP17}.

\begin{theorem}\label{theo:strong:nopolykernel:P3:split}
\textsc{Strong Triadic Closure} has no polynomial compression unless $\textrm{NP} \subseteq \textrm{coNP/ poly}$, even when the input graph is a split graph.
\end{theorem}

\begin{proof}
The reduction comes from the \textsc{Set Packing} problem: given a universe $\mathcal{U}$ of $t$ elements and subsets $B_1, \ldots, B_p$ of $\mathcal{U}$ decide whether there are at least $k$ subsets which are pairwise disjoint.
\textsc{Set Packing} (also known as \textsc{Rank Disjoint Set} problem), parameterized by $|\mathcal{U}|$, does not admit a polynomial compression \cite{DomLS14}.
Given an instance $(\mathcal{U}, B_1, \ldots, B_p, k)$ for the \textsc{Set Packing}, we construct a split graph $G$ with a clique $U \cup Y$ and an independent set $W \cup X$ as follows:

\begin{itemize}
\item The vertices of $U$ correspond to the elements of $\mathcal{U}$.
\item For every $B_i$ there is a vertex $w_i \in W$ that is adjacent to all the vertices of $(U \cup Y) \setminus B_i$.
\item $X$ and $Y$ contain additional $2t$ vertices with $X = \{x_1, \ldots, x_t\}$ and $Y = \{y_1, \ldots, y_t\}$
such that $y_i$ is adjacent to all the vertices of $(W \cup X) \setminus \{x_i\}$ and $x_i$ is adjacent to all the vertices of $(U \cup Y) \setminus \{y_i\}$.
\end{itemize}

Notice that the clique of $G$ contains $2t$ vertices.
We will show that there are at least $k$ pairwise disjoint sets in $\{B_1, \ldots, B_p\}$ if and only if there is a solution for \textsc{Strong $P_3$-closure} on $G$ with at least $k' = |E(U \cup Y)|+(k+t)/2$ edges.

Assume that $\mathcal{B}'$ is a family of $k$ pairwise disjoint sets of $B_1, \ldots, B_p$.
For every $B'_i \in \mathcal{B}'$ we choose three vertices $w_i, y_i, x_i$ from $W$, $Y$, and $X$, respectively, such that $x_i$ is non-adjacent to $y_i$ with the following strong edges: $w_i$ is strongly adjacent to $y_i$ and $x_i$ is strongly adjacent to the vertices of $B'_i$ in $U$.
We also make weak the edges inside the clique between the vertices of $B'_i$ and $y_i$.
All other edges incident to $w_i$ and $x_i$ are weak.
Let $W', Y', X'$ be the set of vertices that are chosen from the family $\mathcal{B}'$ according to the previous description.
Every vertex of $W\setminus W'$ is not incident to a strong edge and, thus, it is isolated in $H$.
For the $t-k$ vertices of $Y\setminus Y'$ we choose a matching
and for each matched pair $y_j,y_{j'}$ we make the following edges strong: $x_jy_{j'}$ and $x_{j'},y_j$ where $x_j$ and $x_{j'}$ are non-adjacent to $y_j$ and $y_{j'}$, respectively.
Moreover each edge $y_jy_{j'}$ of the clique is weak and all other edges incident to $x_j$ and $x_{j'}$ are weak.
The rest of the edges inside the clique $U \cup Y$ are strong.
Figure~\ref{fig:split} illustrates such a labeling on the edges of $G$.

\begin{figure}[t]
\centering
\includegraphics[width=0.5\textwidth]{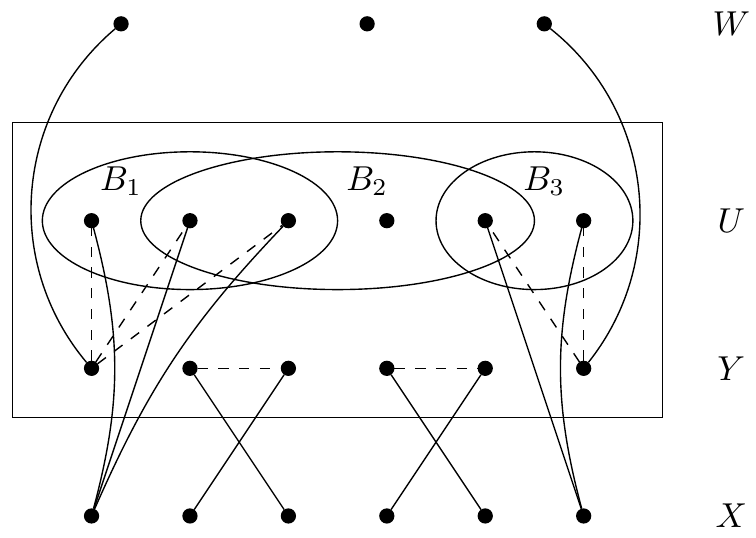}
\caption{Illustrating the split graph $G$ given in the construction in the proof of Theorem~\ref{theo:strong:nopolykernel:P3:split}.
To keep the figure clean, we only draw the strong edges between the independent set $W \cup X$ and the clique $U \cup Y$; the dashed edges of the clique $U \cup Y$ correspond to its weak edges.
Notice that the dashed edges span a union of star graphs.}
\label{fig:split}
\end{figure}

Let us now show that the described subgraph $H$ satisfies the $P_3$-closure with the claimed number of strong edges.
Observe that if there is a $P_3$-graph in $H$ then it must contain a vertex of the independent set incident to a strong edge.
Also notice that no vertex of the clique $U \cup Y$ is strongly adjacent to more than one vertex of the independent set $W \cup X$.
By construction for each $B'_i \in \mathcal{B}'$ the vertices $w_i, x_i$ of the independent set are incident to a strong edge.
The vertices of the clique that are non-adjacent to $w_i$ constitute $B'_i$, and $x_i$ is non-adjacent only to vertex $y_i$.
Since all edges of $E(B'_i,\{y_i\})$ are weak, both vertices $w_i$ and $x_i$ cannot induce a $P_3$-graph.
The rest of the vertices of the independent set that are incident to at least one strong edge belong to $X \setminus X'$.
Every vertex $x_j$ of $X \setminus X'$ is adjacent to all vertices of $(U\cup Y) \setminus\{y_j\}$.
For the strong edge $x_jy_{j'}$ there is a weak edge $y_jy_{j'}$ implying that $x_j$ does not participate in any $P_3$-graph of $H$.
Thus for any vertex $v$ of the independent set that is strongly adjacent to a vertex $v'$ of the clique there are weak edges between $v'$ and the non-neighbors of $v$ in the clique.
Consequently there is no $P_3$-graph in $H$.
For the number of edges in $H$ notice that for every weak edge inside the clique $U \cup Y$ there is a unique matched strong edge incident to a vertex of $X$.
Furthermore every vertex of $W'$ is incident to an unmatched strong edge and each of the half vertices of $X \setminus X'$ is incident to additional unmatched strong edges.
Hence $|E(H)| = |E(U \cup Y)|+ k + (t-k)/2$, which gives the claimed bound $k'$.

For the opposite direction, assume that $H$ is a subgraph of $G$ that satisfies the $P_3$-closure with at least $k'$ edges.
For a vertex $v \in W \cup X$, let $S(v)$ be the strong neighbors of $v$ in $H$ and let $B(v)$ be the non-neighbors of $v$ in $U \cup Y$.
Our task is to show that for any two vertices $u,v$ of $W \cup X$ with non-empty sets $S(u), S(v)$, we have $B(u) \cap B(v)=\emptyset$.
Since there is no $P_3$-graph in $H$, it is clear that all edges of $E(S(v),B(v))$ are weak.
Also observe that for any two vertices $u,v \in W \cup X$, $S(u)\cap S(v) = \emptyset$.
We first prove the following important property for the weak edges of the clique $U \cup Y$.

\begin{claim}\label{claim:weakedgeclique}
An edge $uv$ of the clique $U \cup Y$ does not belong to $H$ only if there are two vertices $u',v'$ in the independent set $W \cup X$ such that $u \in S(u') \cap B(v')$ and $v \in S(v') \cap B(u')$.
\end{claim}

\begin{claimproof}
We will show that in all other cases we can safely make the edge $uv$ strong and maintain the same number of strong edges.
If there are no such vertices in $W \cup X$ then we can make the edge $uv$ strong without violating the $P_3$-closure.
Thus there is at least one vertex $u'$ that is strongly adjacent to $u$ so that $u\in S(u')$.
Moreover if $v \in S(u')$ then both $u$ and $v$ have no other strong neighbor in the independent set which means that we can safely make the edge $uv$ strong.
This implies that $v \notin S(u')$.
Now assume that $v \notin B(u')$, meaning that $v$ is a neighbor of $u'$ in $G$ but not a neighbor of $u'$ in $H$.
Observe that $v$ has at most one strong neighbor in the independent set.
If there is such a strong neighbor $v'$ of $v$ in $W \cup X$ then we make $vv'$ weak and $uv$ strong.
Such a replacement is safe, since $u$ has exactly one strong neighbor $u'$ in $W \cup X$ and all other strong neighbors of $u$ or $v$ belong to the clique.
Hence $v \in B(u')$.

Suppose next that $v$ has no strong neighbor in the independent set.
Then we replace the strong edge $u'u$ by the edge $uv$; such a replacement is safe since $v$ has no strong neighbors in the independent set and $u'$ is the only strong neighbor of $u$ in the independent set.
Thus there is a strong neighbor $v'$ of $v$ such that $v' \in W \cup X$.
Summarizing, there are $u',v' \in W \cup X$ such that $u \in S(u')$, $v \in S(v')$, $v \in B(u')$, and by symmetry reasons for $v'$ we get $u \in B(v')$.
Therefore $u \in S(u') \cap B(v')$ and $v \in S(v') \cap B(u')$.
\end{claimproof}

We next consider the vertices of $X$ from the independent set.

\begin{claim}\label{claim:nostrongW}
Let $H$ be solution in which no vertex of $W$ is incident to a strong edge.
Then $|E(H)| \leq |E(U \cup Y)| + t/2$.
\end{claim}
\begin{claimproof}
We first show that for every vertex $x_i$ of $X$, $S(x_i)$ contains at most one vertex.
Recall that $B(x_i)$ contains exactly one vertex.
Assume for contradiction that $S(x_i)$ contains at least two vertices.
Let $u,v\in S(x_i)$ and let $B(x_i)=z$.
By the $P_3$-closure, both edges $uz$ and $vz$ of the clique must be weak.
Then applying Claim~\ref{claim:weakedgeclique}, there is a vertex $x_j \in X$ such that $z \in S(x_j)$ and $\{u,v\} \subseteq B(x_j)$.
This however is not possible since by construction we know that $B(x_j)$ contains exactly one vertex.
Thus $|S(x_i)|\leq 1$ for every vertex $x_i \in X$.

Let $E_W$ be the set of weak edges that have both their endpoints in the clique.
If there are two edges of $E_W$ incident to the same vertex $u$ then by Claim~\ref{claim:weakedgeclique} the unique vertex $u' \in X$ that is strongly adjacent to $u$ has two non-adjacent vertices in the clique.
Since every vertex of $X$ is non-adjacent to exactly one vertex, there are no two edges of $E_W$ incident to the same vertex.
This means that the edges of $E_W$ form a matching in $E(U \cup Y)$.
Moreover by Claim~\ref{claim:weakedgeclique} for every edge of $E_W$ there are exactly two strong edges between the vertices of the independent set and the clique.
Thus $E(H) = (E(U \cup Y) \setminus E_W) \cup S(X)$ where $S(X)$ are the strong edges with one endpoint in $X$.
Therefore $|E(H)| \leq |E(U \cup Y)| - t/2 + t$, since there are $t$ vertices in $X$.
\end{claimproof}

Thus by Claim~\ref{claim:nostrongW} and the fact that $H$ contains $k'>|E(U \cup Y)| + t/2$ edges, we know that some vertices of $W$ are incident to strong edges in an optimal solution $H$.
We next show that these type of vertices of $W$ must have disjoint non-neighborhood in $G$.
To do so, we consider the {\it weak components} of $E(H)$ in the clique.
A {\it weak component} is a connected component of the clique spanned by the weak edges, that is, by the edges of $E(G)-E(H)$.

Let $C_w$ be a weak component with $n(C_w)$ its number of vertices and $m(C_w)$ its number of weak edges.
Denote by $E_S(C_w)$ the set of strong edges incident to a vertex of $C_w$ and a vertex of the independent set $W \cup X$.
By Claim~\ref{claim:weakedgeclique}, $E_S(C_w)$ is non-empty.
Then the number of edges in $H$ can be described as follows:
$$
|E(H)| = |E(U \cup Y)| + \sum_{C_w} \left(|E_S(C_w)| - m(C_w)\right).
$$
Notice that every vertex of $C_w$ has exactly one strong neighbor in the independent set, since $H$ satisfies the $P_3$-closure.
This means that $|E_S(C_w)| = n(C_w)$.
If we make strong all edges among the vertices of $C_w$ and remove the edges of $E_S(C_w)$ from $H$ then the resulting graph satisfies the $P_3$-closure. %,
Thus if $m(C_w) \geq n(C_w)$ then we can safely ignore such a component in the sum of $|E(H)|$ by replacing all its weak edges by the strong edges of $E_S(C_w)$.
This means that $m(C_w)=n(C_w)-1$ because the weak edges of $C_w$ span a connected component.
Therefore every weak component $C_w$ spans a tree in $H$ and $|E(H)| = |E(U \cup Y)| + c$, where $c$ is the number of weak components in $H$.

In fact we will prove that every weak component spans a tree of height one (\emph{star graph}).
Before that, let us first show the following property with respect to the nested non-neighborhood of vertices of $W \cup X$.
For a vertex $v \in W \cup X$, observe that all edges between $S(v)$ and $B(v)$ are weak.
Thus all vertices of $S(v)$ belong to the same weak component of $H$.

\begin{claim}\label{claim:subsetsB}
Let $v_i,v_j \in W \cup X$ such that both $S(v_i)$ and $S(v_j)$ are non-empty.
Then $B(v_i) \not\subseteq B(v_j)$ and $B(v_j) \not\subseteq B(v_i)$.
\end{claim}
\begin{claimproof}
By the $P_3$-closure of $H$, we know that no vertex of the clique has more than one strong neighbor in the independent set, which implies $S(v_i) \cap S(v_j) = \emptyset$.
Assume for contradiction that $B(v_i) \subseteq B(v_j)$.
This means that the vertices of $S(v_i) \cup S(v_j)$ belong to the same weak component $C_w$.
We show that there is an optimal solution $H'$ for which $S(v_j) = \emptyset$ and $|E(H')|=|E(H)|$.
There is no weak edge with the endpoints in $S(v_i)$ and $S(v_j)$, respectively, since $C_w$ is a tree.
Thus all edges between the vertices of $S(v_i)$ and $S(v_j)$ are strong.
This means that $v_i$ is adjacent to every vertex of $S(v_j)$.
We construct $H'$ by replacing all strong edges incident to $v_j$ by strong edges incident to $v_i$.
Remove all strong edges incident to $v_j$ and let $S(v_i) \cup S(v_j)$ be the strong neighbors of $v_i$ in $H'$.
Notice that $|E(H')|=|E(H)|$.
Since we only added strong edges incident to $v_i$ and $B(v_i) \subseteq B(v_j)$, all edges between $B(v_i)$ and $S(v_i) \cup S(v_j)$ are weak and, thus, $H'$ satisfies the $P_3$-closure.
Therefore applying the same replacement for every pair of vertices with nested non-neighborhood, results in an optimal solution as desired.
\end{claimproof}

Let us now show that every weak component of $H$ spans a star graph.
Let $u_1, u_2, \ldots, u_r$ be a path of a weak component $C_w$ of $H$ where $u_1$ is a leaf vertex of $C_w$.
Since $u_1u_2$ is a weak edge, Claim~\ref{claim:weakedgeclique} implies that there is a vertex $v_i$ in the independent set that is strongly adjacent to $u_1$ such that $B(v_i)=\{u_2\}$.
If $C_w$ is not a star then $r\geq 4$.
For $r\geq 4$, we know that there is a vertex $v_j$ in the independent set that is strongly adjacent to $u_3$ such that $B(v_j)=\{u_2,u_4\}$.
Then we reach a contradiction since $B(v_i) \subset B(v_j)$ which is not possible by Claim~\ref{claim:subsetsB}.
Therefore $r\leq 3$ which implies that $C_w$ spans a tree of height one.

\begin{claim}\label{claim:disjointB}
Let $v_i,v_j \in W \cup X$ such that both $S(v_i)$ and $S(v_j)$ are non-empty.
Then $B(v_i) \cap B(v_j) = \emptyset$.
\end{claim}
\begin{claimproof}
Recall that $S(v_i) \cap S(v_j) = \emptyset$ and notice that all edges of $E(S(v_i),B(v_i))$ and $E(S(v_j),B(v_j))$ are weak.
If the vertices of $S(v_i)$ belong to a different weak component from the vertices of $S(v_j)$ then $B(v_i)$ and $B(v_j)$ are disjoint.
Suppose that the vertices of $S(v_i)$ and $S(v_j)$ belong to the same weak component $C_w$.
Let $u$ be the non-leaf vertex of the star $C_w$.
If both $v_i$ and $v_j$ are strongly adjacent to leaf vertices of $C_w$, then $B(v_i)=B(v_j)=\{u\}$.
Thus by Claim~\ref{claim:subsetsB}, $v_i$ is strongly adjacent to $u$ so that $B(v_i) = C_w \setminus\{u\}$ and $v_j$ is strongly adjacent to all leaf vertices of $C_w$ so that $B(v_j)=\{u\}$.
Consequently $B(v_i)$ and $B(v_j)$ are disjoint sets.
\end{claimproof}

By construction, $B(x)$ with $x\in X$ is disjoint with any $B(v)$ of a vertex $v \in (W \cup X)\setminus\{x\}$.
Claims~\ref{claim:nostrongW} and \ref{claim:disjointB} imply that all vertices of $W$ that are incident to at least one strong edge in $H$ must have disjoint non-neighborhood.
Since $B(w_i)=B_i$, there are $k$ pairwise disjoint sets in $\{B_1, \ldots, B_p\}$ for the $k$ vertices of $W$ that are incident to at least one strong edge in $H$.
Therefore there is a solution for the \textsc{Set Packing} problem for $(\mathcal{U}, B_1, \ldots, B_p, k)$.
\end{proof}

Let $F$ be a graph that has at least one component with at least three vertices. If $M$ is a matching in a graph $G$, then the spanning subgraph $H$ of $G$ with $E(H)=M$ satisfies the $F$-closure. It implies that an instance $(G,k)$ of \textsc{Strong $F$-Closure} is a yes-instance of the problem if the maximum matching size $\mu(G)\geq k$. Since  a maximum matching can be found in polynomial time~\cite{MicaliV80}, we can solve  \textsc{Strong $F$-Closure} in polynomial time for such instances. This gives rise to the question about the parameterized complexity of \textsc{Strong $F$-Closure} with the parameter $r=k-\mu(G)$. We show that \textsc{Strong Triadic Closure} is FPT with this parameter for the instances where $\Delta(G) \le 4$. Note that
{\sc Strong Triadic Closure} is NP-complete on graphs $G$ with $\Delta(G) \le d$ for every $d\geq 4$ \cite{KonstantinidisN17}.

\begin{theorem}\label{theo:strong:bounded4:fpt}
{\sc Strong  Triadic Closure} can be solved in time $2^{\Oh(r)}\cdot n^{\Oh(1)}$ on graphs of maximum degree at most 4, where $r=k-\mu(G)$.
\end{theorem}

\begin{proof}
Let $(G,k)$ be a instance of {\sc Strong  Triadic Closure} such that $\Delta(G)\leq 4$. Let also $r = k-\mu(G)$.

We construct the set of vertices $X$ and the set of edges $A$ as follows. Initially, $X=\emptyset$ and $A=\emptyset$. Then we exhaustively perform the following steps in a greedy way:
\begin{enumerate}
\item If there exists a copy of  $K_4$ in $G-X$, we add the vertices of this $K_4$ to $X$ and the edges between these vertices to $A$.
\item If there exists a triangle $T$ in $G-X$ such that $\mu(G-X)<3+\mu(G-X-T)$, we add the vertices of $T$ to $X$ and and the edges of $T$ to $A$.
\end{enumerate}

Let $M$ be a maximum matching of $G-X$ for the obtained set $X$. Note that the spanning subgraph $H$ of $G$ with the set of edges $A\cup M$ is a disjoint union of complete graphs with 1, 2, 3 or 4 vertices, that is, $H$ has no induced path on three vertices. Hence, $H$ satisfies the $P_3$-closure.
Assume that Step~1 was applied $p$ times and we used Step~2 $q$ times. Clearly, $|A|=6p+3q$.
Notice that the vertices of a copy of $K_4$ can be incident to at most 4 edges of a matching and  the complete graph with 4 vertices has 6 edges. Observe also that by the application of Step~2, we increase the size of $A$ by 3 and $\mu(G-X)-\mu(G-X-T)\leq 2$. This implies that
$|E(H)|=|A|+|M|\geq \mu(G)+2p+q$.
Therefore, if $2p+q\geq r$, $(G,k)$ is a yes-instance of {\sc Strong  Triadic Closure}. Assume from now that this is not the case. In particular, it means that $|X|\leq 4r$ and $G'=G-X$ is a $K_4$-free graph. By the choices made in both steps, notice that every vertex of $X$ has at least two neighbors inside $X$.
Let $Y=V(G)\setminus X=V(G')$.

We need some structural properties of $G'$ and (possible) solutions for the considered instance of {\sc Strong  Triadic Closure}.

\begin{claim}\label{claim:trianglesM}
If $T$ is a triangle in $G'$, then $T$ satisfies the following properties:
\begin{enumerate}
\item[(i)] $T$ contains no edge of $M$;
\item[(ii)] every vertex of $T$ is incident to an edge of $M$.
\end{enumerate}
\end{claim}

\begin{claimproof}
In both cases, the triangle $T$ is such that at most two edges of the matching~$M$ are incident to its vertices. This implies that $\mu(G')<3+\mu(G'-T)$, which is a contradiction with the fact that Step~2 can no longer be applied.
\end{claimproof}

We say that a solution $H$ for $(G,k)$ is \emph{regular} if $H[Y]$ is a disjoint union of triangles, edges and isolated vertices.
We also say that a solution $H$ is \emph{triangle-maximal} if (i) it contains the maximum number of edges and, subject to (i), (ii) contain the maximum number of pairwise distinct triangles.%i) the 

\begin{claim}\label{claim:structureG-X}
If $(G,k)$ is a yes-instance of {\sc Strong  Triadic Closure}, then every triangle-maximal solution is regular.
\end{claim}

\begin{claimproof}
Let $H$ be a triangle-maximal solution for $(G,k)$.

We first note that, $H'$ has no $K_{1,3}$ as a subgraph. Otherwise it would imply the existence of a $K_4$ in $G-X$, because for every copy $xyz$ of (not necessarily induced) $P_3$ in $H$, $xz\in E(G)$ if $H$ satisfies the $P_3$-closure.
This implies that $H'$ consists of a disjoint union of paths and cycles.
Consider an induced path on three vertices $P_3=v_1v_2v_3$ in $H'$.
By the $P_3$-closure there is the edge $v_1v_3$ in $G$.
We prove that the $P_3$ has a particular form which allows us to make $v_1v_3$ strong, i.e., the triangle $v_1v_2v_3$ belongs to a solution $H$.
In particular we show that $N_H(v_1)=\{v_2\}$, $N_H(v_2)=\{v_1,v_3\}$, and $N_H(v_3)=\{v_2,y\}$ where $y$ is a vertex in $G$.
\begin{itemize}
\item First observe that $v_2$ has no other neighbor in $H'$, because $v_2$ belongs to a path or a cycle in $H'$.
Assume that there is a vertex $x\in X$ that is adjacent to $v_2$ in $H$.
Then by the $P_3$-closure, $x$ is adjacent in $G$ to all three vertices of $P_3$ which contradicts the fact that $d(x)\leq 4$,
because $x$ is adjacent to at least two vertices inside $X$.
Thus $v_2$ has no other neighbor in $H$.

\item Next assume that there are vertices $u_1,u_3$ such that $u_1 \in N_H(v_1)\setminus\{v_2\}$ and $u_3 \in N_H(v_3)\setminus\{v_2\}$.
If $u=u_1=u_3$ then $u$ does not belong to $Y$ because there is no $K_4$ in $G'$.
And if $u\in X$ then by the $P_3$-closure, $u$ is adjacent to all three vertices of the $P_3$ which contradicts the fact that $d(u) \leq 4$.
For $u_1\neq u_3$, notice that $v_2$ is adjacent to both $u_1,u_3$ by the $P_3$-closure.
Then both $u_1v_1v_2$ and $u_3v_3v_2$ form triangles in $G$, which implies by Claim~\ref{claim:trianglesM} that there is an edge $v_2v$ of $M$ with $v\notin \{v_1,v_3,u_1,u_3\}$.
This, however, contradicts the fact that $d(v_2)\leq 4$.

\item By the previous two arguments, we know that at least one of $v_1,v_3$ is only adjacent to $v_2$ in $H$.
Without loss of generality, assume that $N_H(v_1)=\{v_2\}$.
If $N_H(v_3)=\{v_2,y,y'\}$ then by the $P_3$-closure $v_2$ is adjacent in $G$ to both $y,y'$.
Applying Claim~\ref{claim:trianglesM} shows that there is another edge incident $v_2$, contradicting the fact that $d(v_2)\leq 4$.
Also note that if $N_H(v_3)=\{v_2\}$ then both $v_1,v_3$ have no other strong edge incident to them, so that the edge $v_1v_3$ of $G$ can be made strong which contradicts the maximality of $H$.
\end{itemize}
Thus for the given $P_3$ we know that $N_H(v_1)=\{v_2\}$, $N_H(v_2)=\{v_1,v_3\}$, and $N_H(v_3)=\{v_2,y\}$.
This means that we can replace in $H$ the edge $v_3y$ by the edge $v_1v_3$ without violating the $P_3$-closure.
Iteratively applying such a replacement for every $P_3$ of $H'$ shows that $H$ is regular.
\end{claimproof}

\begin{claim}\label{claim:triangled1}
Let $T=abc$ be a triangle in $G'$ that is at distance one from $X$. If $H$ is a solution containing $T$, then $H$ contains no other edge incident to the vertices $a$, $b$, and $c$.
\end{claim}

\begin{claimproof}
Let $T=abc$ be a triangle as described above and let $H$ be a solution containing~$T$.
Assume for a contradiction that there exists an edge $xa$ in $H$ that is incident to a vertex of $T$.
Suppose that $x\in X$. 
This implies that $xb\in E(G)$ and $xc\in E(G)$.
Since $x$ has at least two neighbors inside $X$, we conclude that $d(x)>4$, a contradiction.
If $x \in G-X$, this would imply the existence of $K_4$ in $G'$, a contradiction.
\end{claimproof}

\begin{claim} \label{claim:1triangle}
Let $T$ be a triangle at distance at least two from $X$ that does not intersect any other triangle. Then $T$ is included in every triangle-maximal regular solution for $(G,k)$.
\end{claim}

\begin{claimproof}
Let $T=abc$ be a triangle as described above and assume that $H$ is a triangle-maximal regular solution
that does not contain $T$. Since no other triangle intersects $T$, at most one edge of $H$ is incident to each vertex of $T$ by Claim~\ref{claim:structureG-X}.
If no edge of $T$ is in $H$, we can replace the edges incident in $T$ by the edges $ab$, $bc$ and $ac$ and obtain a solution with at least as many edges as $H$ containing~$T$.
If there exists an edge of $T$ in $H$, let $ab$ be such an edge.
Again, since no other triangle intersects $T$, there is no other edge of the solution that is incident to $a$ or $b$ and at most one edge is incident to $c$.
Then we replace the edge incident to $c$ by the two edges of the triangle $abc$ and obtain a solution with more edges, a contradiction.
\end{claimproof}

\begin{claim} \label{claim:2triangles}
If $T_1$ and $T_2$ are two intersecting triangles in $G'$, then the following holds:
\begin{enumerate}
\item $T_1$ and $T_2$ have one edge in common;
\item No other triangle intersects $T_1$ or $T_2$;
\end{enumerate}
\end{claim}
\begin{claimproof}
Let $T_1$ and $T_2$ be two intersecting triangles as described above. Assume for a contradiction that $T_1$ and $T_2$ have only a single vertex in common and let $a$ be such a vertex. Recall that $M$ is a maximum matching in $G-X$. By Claim~\ref{claim:trianglesM}, there exists an edge of the matching incident to $a$ that cannot be contained neither in $T_1$ nor in $T_2$, which implies that $d(a)>4$, which is a contradiction. We conclude that the triangles must intersect in one edge. Let $T_1=abc$ and $T_2=bcd$. Assume for the sake of contradiction that there exists a triangle $T$ that intersects $T_1$. Since by the first argument of the claim the triangles $T$ and $T_1$ cannot intersect in a single vertex, $T$ contains at least one of $b$ or $c$. Assume $b\in V(T)$. Again, by Claim~\ref{claim:trianglesM}, there must be an edge of $M$ incident to $b$ that is not contained in any of the triangles, which implies that $d(b)>4$, a contradiction. This concludes the proof.
\end{claimproof}

\begin{claim}\label{claim:intersecting}
If $T_1$ and $T_2$ are two intersecting triangles such that $T_1$ is at distance at least two from $X$, then either $T_1$ or $T_2$ is  included in every triangle-maximal regular solution for $(G,k)$.
\end{claim}

\begin{claimproof}
Let $H$ be a solution. By Claim~\ref{claim:2triangles}, $T_1$ and $T_2$ have exactly one common edge. Let $T_1=abc$ and $T_2=bcd$.
Assume that a triangle-maximal regular solution $H$
contains neither $T_1$ nor $T_2$. Note that at most one edge of $H$ is incident to $a$ by Claim~\ref{claim:structureG-X}.
Because $H$ does not contain all the edges of $T_2$, the same holds for $b$ and $c$ by  Claim~\ref{claim:structureG-X}. Now we repeat the same arguments as in the proof of Claim~\ref{claim:1triangle}. If no edge of $T_1$ is in $H$, we can replace the edges incident in $T_1$ by the edges $ab$, $bc$ and $ac$ and obtain a solution with at least as many edges as $H$ and containing $T_1$. If there exists an edge of $T_1$ in $H$, then at most two edges of $H$ are incident to the vertices of $T_1$ and we can replace them by the edges of $T_1$ and increase the number of edges in the solution contradicting the choice of $H$.
\end{claimproof}

\medskip
Given the properties of the triangles in $G'$ and the properties of triangle-maximal regular solutions, we can now ready to solve the problem by finding a regular solution if it exists.
Recall that by Claim~\ref{claim:structureG-X}, a regular solution $H$ to the problem when restricted to $G-X$ is a disjoint union of triangles, edges and isolated vertices.
The crucial step is to sort out triangles in $G'$.

We first consider the triangles in $G'$ that are at distance at most one from the set~$X$ in $G$, that is, the triangles that contain at least one vertex that is adjacent to a vertex of $X$ in $G$. Since $|X|\leq 4r$ and since every vertex of $X$ has at least two neighbors inside $X$, we have that $|N_{G}(X)|\leq 8r$. By Claim~\ref{claim:2triangles}, at most 2 triangles of $G'$ contain the same vertex. Thus, the number of pairwise distinct triangles in $G'$ that are at distance at most one from the set~$X$ in $G$ is at most $16r$. We list all these triangles, and branch on all at most $2^{16r}$ choices of the triangles that are included in a triangle-maximal regular solution. Then, for each choice of these triangles, we try to extend the partial solution. If we obtain a solution for one of the choices we return it and the algorithm returns NO otherwise.

Assume that we are given a set $\mathcal{T}_1$ of triangles at distance one from $X$ that should be a solution. Note that by Claim~\ref{claim:structureG-X}, the triangles in $\mathcal{T}_1$ are pairwise disjoint. We apply the following reduction rule.

\begin{Rule} \label{rule:delTriangles-one}
Set $G=G-\cup_{T\in\mathcal{T}_1}T$ and set $k=k-3|\mathcal{T}_1|$.
\end{Rule}

By Claim~\ref{claim:triangled1}, the original instance has a regular solution if and only if the obtained instance has a regular solution that does not contain triangles in $G-X$ that are at distance one from $X$. Our aim now is to find such a solution. For simplicity, we keep the same notation and assume that $G'=G-X$.

Now we deal with triangles that are at distance at least 2 from $X$.
Consider the set $\mathcal{T}_2$ of triangles in $G'$ that are at distance at least 2 from $X$ and have no common vertices with other triangles in $G'$. By Claim~\ref{claim:1triangle}, all these triangles are in every triangle-maximal regular solution.
It immediately gives us the following rule.

\begin{Rule} \label{rule:delTriangles-two}
Set $G=G-\cup_{T\in\mathcal{T}_2}T$ and set $k=k-3|\mathcal{T}_2|$.
\end{Rule}

We again assume that $G'=G-X$.
To consider the remaining triangles, recall that by Claim~\ref{claim:2triangles}, for every such a triangle $T$, $T$ is intersecting with a unique triangle $T'$ of $G'$ and $T,T'$ are sharing an edge.

Let $\mathcal{T}_3$ be the set of triangles in $G'$ that are at distance at least 2 from $X$ in $G$ and have a common edge with a triangle at distance one from $X$.
Recall that we are looking for  a regular solution that does not contain triangles in $G-X$ that are at distance one from $X$. Then by
Claim~\ref{claim:intersecting}, triangles of $\mathcal{T}_3$ should be included to a triangle-maximal regular solution,
and we get the next rule.

\begin{Rule} \label{rule:delTriangles-three}
Set $G=G-\cup_{T\in\mathcal{T}_3}T$ and set $k=k-3|\mathcal{T}_3|$.
\end{Rule}

As before, let $G'=G-X$. The remaining triangles in $G'$ at distance at least 2 from $X$ in $G$ form pairs $\{T_1,T_2\}$ such that $T_1$ and $T_2$ have a common edge and are not intersecting any other triangle. Let $\mathcal{P}$ be the set of all such pairs. By Claim~\ref{claim:intersecting}, a triangle-maximal regular solution
contains either $T_1$ or $T_2$. We use this to apply the following rule.

\begin{Rule} \label{rule:delPairs}
For every pair $\{T_1,T_2\}\in\mathcal{P}$, delete the vertices of $T_1$ and $T_2$ from $G$, construct a new vertex $u$ and make it adjacent to the vertices of $N_G((T_1\setminus T_2)\cup(T_2\setminus T_1))$.
Set $k=k-3|\mathcal{P}|$.
\end{Rule}

Denote by  $(\hat{G},\hat{k})$ the instance of {\sc Strong  Triadic Closure} obtained from $(G,k)$ by the application of Rule~\ref{rule:delPairs}.
We show the following claim.

\begin{claim}\label{claim:delPairs-one}
If  the instance  $(G,k)$ has a triangle-maximal regular solution $H$ that has no triangles in $G-X$ at distance one from $X$, then  there is a solution $\hat{H}$ for $(\hat{G},\hat{k})$ such that $\hat{H}-X$ is a disjoint union of edges and isolated vertices, and if there is a solution $\hat{H}$ for $(\hat{G},\hat{k})$ such that $\hat{H}-X$ is a disjoint union of edges and isolated vertices, then $(G,k)$ has a regular solution $H$ that has no triangles in $G-X$ at distance one from $X$
 \end{claim}

\begin{claimproof}
Let $H$ be a triangle-maximal regular solution for $(G,k)$ such that $H$ has no triangles in $G-X$ at distance one from $X$. Notice that if $H$ contains a triangle, then it belongs to one of the pairs of $\mathcal{P}$. By Claim~\ref{claim:intersecting}, we can assume that $H$ contains a triangle from every pair from $\mathcal{P}$.
 We construct a solution $\hat{H}$ for $(\hat{G},\hat{k})$ by modifying $H$ as follows. First, we include in $\hat{H}$ the edges of $H$ that are not incident to the vertices of the pairs of triangles of $\mathcal{P}$.
 For every pair $\{T_1,T_2\}\in\mathcal{P}$, $H$ contains either $T_1$ or $T_2$. Assume without loss of generality that $T_1$ is in $H$. Let $v$ be the vertex of $T_2$ that is not included in $T_1$. By Claims~\ref{claim:structureG-X} and \ref{claim:trianglesM}, at most one edge of $H$ is incident to $v$ and there is no edge in $H$ that is incident to exactly one vertex of $T_1$. Let $u$ be the vertex of $\hat{G}$ constructed by Rule~\ref{rule:delPairs} for $\{T_1,T_2\}$. If $vx\in E(H)$ for some $x\in V(G)$, then we include the edge $ux'$ in $\hat{H}$, where $x'$ is the vertex constructed from $x$ by the rule; note that it can happen that $x$ is a vertex of some other pair of triangles.
 Since we include in $\hat{H}$ at most one edge incident to a vertex constructed by the rule, $\hat{H}$ does not contain triangles and is a disjoint union of edges and isolated vertices.
 Moreover, since $|E(H)|\geq k$, we have that $|E(\hat{H})|\geq k-3|\mathcal{P}|=\hat{k}$.

Suppose now that $\hat{H}$ is a solution for $(\hat{G},\hat{k})$ such that $\hat{H}-X$ is a disjoint union of edges and isolated vertices. Now we construct $H$ by modifying $\hat{H}$. For every edge $uv$ of $\hat{H}$ such that $u$ and $v$ are vertices of the original graph $G$, we include $uv$ in $H$. Assume that $uv\in E(\hat{H})$ is such that $v\in V(G)$ and $u$ was obtained from a pair  $\{T_1,T_2\}\in \mathcal{P}$. Then $v$ is adjacent in $G$ to a vertex $x$ that belongs to exactly one of the triangles, say $T_1$. We include $xv$ and $T_2$ in $H$.
Suppose that $uv\in E(\hat{H})$ is such that $u$ was obtained from a pair  $\{T_1,T_2\}\in \mathcal{P}$ and $v$ was obtained from a pair $\{T_1',T_2'\}\in \mathcal{P}$.
Then $G$ has an edge $xy$ such that  $x$ that belongs to exactly one of the triangles  $T_1,T_2$, say $T_1$, and  $y$  belongs to exactly one of the triangles  $T_1',T_2'$, say $T_1'$.
We include $xy$, $T_2$ and $T_2'$ in $H$. Finally, if there is a pair $\{T_1,T_2\}\in \mathcal{P}$ such that for the vertex $u\in V(\hat{G})$ constructed from this pair, $\hat{H}$ has no edge incident to $u$, we include $T_1$ in $H$. With this way we obtain $H$ such that $H-X$ is a disjoint union of triangles, edges and isolated vertices. It remains to note that because  $|E(\hat{H})|\geq \hat{k}$, we have that $|E(H)|\geq k$, that is, $H$ is a regular solution.
\end{claimproof}

By Claim~\ref{claim:delPairs-one}, we have to find a solution for the instance  $(\hat{G},\hat{k})$ such that $\hat{H}-X$ is a disjoint union of edges and isolated vertices. We do it by branching on all possible choices of edges in a solution that are incident to the vertices of $X$. Since $|X|\leq 4$ and $\Delta(G)\leq 4$, there are at most $16r$ edges that are incident to the vertices of $X$ and, therefore, we branch on at most $2^{16r}$ choices of a set of edges $S$. Then for each choice of $S$, we are trying to extend it to a solution. If we can do it for one of the choices, we return the corresponding solution, and the algorithm  returns NO otherwise.

Assume that $S$ is given. First, we verify whether the spanning subgraph of $G$ with the set of edges $S$ satisfies the $P_3$-closure. If it is not so, we discard the current choice of $S$ since, trivially, $S$ cannot be extended to a solution. Assume that this is not the case. Let $R=\hat{G}-X$. We modify $R$ by the exhaustive application of the following rule.

\begin{Rule} \label{rule:delR}
If there is $xy\in E(R)$ such that there is $z\in X$ such that
$xz\in S$ and $yz\notin E(\hat{G})$, then delete $xy$ from $R$.
\end{Rule}

Let $R'$ be the graph obtained from $R$ by the rule. Observe that the edges deleted by Rules~\ref{rule:delR} cannot belong to a solution. Hence, to extend $S$, we have to complement it by some edges of $R'$ that form a matching. Moreover, every matching of $R'$ could be used to complement $S$. Respectively, we find a maximum matching $M$ in $R'$ in polynomial time~\cite{MicaliV80}. We obtain that the spanning subgraph $\hat{H}$ of $\hat{G}$ with $E(\hat{H})=S\cup M$ satisfies the $P_3$-closure. We verify whether $|S|+|M|\geq \hat{k}$. If it holds, we return $\hat{H}$. Otherwise, we discard the current choice of $S$.

The correctness of the algorithm follows from the properties of Rules~\ref{rule:delTriangles-one}--\ref{rule:delR} and Claim~\ref{claim:delPairs-one}. To evaluate the running time, observe that Steps~1 and 2 that were used to construct $X$ and $A$ can be done in polynomial time. Then we branch on at most $2^{16r}$ choices of $\mathcal{T}_1$. For each choice, we apply Rules~\ref{rule:delTriangles-one}--\ref{rule:delPairs} in polynomial time. Then we consider at most $2^{16r}$ choices of a set of edges $S$. For each choice, we apply Rule~\ref{rule:delR} in polynomial time and then compute a maximum matching in $R'$~\cite{MicaliV80}. Summarizing, we obtain the running time $2^{\Oh(r)}\cdot n^{\Oh(1)}$.
\end{proof}

%%%%%%%%%%%%%%%%%%%%%%%%%%%%%%%%%%%%%%%%%%%%%%%%%%%%%%%%%%%%%%%%%%%%%%%%%%%%%%%%%%%%%%%%%%%%
\section{Concluding remarks}
\label{sec:conc}
To complement our results so far, we give here the parameterized complexity results when our problem is parameterized by the number of weak edges. The following result is not difficult to deduce using similar ideas to those used in proving that \textsc{$F$-free Edge Deletion} is FPT by the number of deleted edges \cite{Cai96}.

\begin{theorem}\label{theo:weak:bounded:fpt}
For every fixed graph $F$, \textsc{Strong $F$-closure} can be solved in time $2^{\Oh(\ell)}\cdot n^{{O}(1)}$, where $\ell= |E(G)|-k$.
\end{theorem}

\begin{proof}
 We basically use the main idea given in \cite{Cai96}.
 Since $F$ is of fixed size, we can list all its induced subgraphs isomorphic to $F$ in polynomial time. 
 For each induced subgraph $F'$ we check whether $G[F'] \simeq F$.
 If $G[F'] \simeq F$, then we must remove at least one of the edges of $F'$.
 We branch at all such possible edges and on each resulting graph we apply the same procedure for at most $\ell$ steps.
 If at some intermediate graph we have $G[F'] \not\simeq F$ for all of its induced subgraphs then we have found the desired subgraph within at most $\ell$ edge deletions.
 Otherwise, we can safely output that there is no such subgraph with at most $\ell$ edge removals.
 The main difference with the algorithm given in \cite{Cai96}, is that we first check at each intermediate graph whether the vertices of $F'$ induce a forbidden graph in $G$, before branching at each of the subgraphs.
 As the branching will generate at most $|E(F)|$ such instances  and the depth of the search tree is bounded by $\ell$, the overall running time is  $2^{\Oh(\ell)}\cdot n^{\mathcal{O}(1)}$.
\end{proof}

Next we show that \textsc{Strong $F$-closure} has a polynomial bi-kernel with this parameterization whenever $F$ is a fixed graph.
We obtain this result by constructing bi-kernelization that reduces \textsc{Strong $F$-closure} to the \textsc{$d$-Hitting Set} problem that is the variant of \textsc{Hitting Set} with all the sets in $\mathcal{C}$ having $d$ elements.
Notice that this result comes in contrast to the \textsc{$F$-free Edge Deletion} problem, as it is known that there are fixed graphs $F$ for which there is no polynomial compression~\cite{CC15} unless $\textrm{NP}\subseteq \textrm{coNP/ poly}$.

\begin{theorem}\label{theo:weak:bounded:polykernel}
For every fixed graph $F$, \textsc{Strong $F$-closure} has a polynomial bi-kernel, when parameterized by $\ell=
|E(G)|-k$.
\end{theorem}

\begin{proof}
Let $d$ be the number of edges of $F$.
We enumerate all the induced subgraphs of $G$ isomorphic to $F$ in polynomial time.
Let $\mathcal{F}_G=\{F_1, \ldots, F_q\}$ be the produced subgraphs isomorphic to $F$ such that $V(F_i) \neq V(F_j)$.
For each $F_i \in \mathcal{F}_G$, we construct the set $E_i = E(F_i)$.
Notice that $|E_1|= \cdots = |E_q|=d$.
Now our task is to select at most $\ell$ edges $E'$ from $G$ such that $E' \cap E_i \neq \emptyset$ for every $E_i$.
We claim that such a subset of edges is enough to produce a solution for the \textsc{Strong $F$-closure}.
To see this, consider an $F$-graph $F_i$ of $G$ and denote by $G'$ the graph obtained from $G$ by removing an edge $e=xy$ of $F_i$.
Assume for contradiction that at least one new $F$-graph $F'$ is created in $G'$ so that $F' \notin \mathcal{F}_G$ and $F' \in \mathcal{F}_{G'}$.
Then both $x$ and $y$ must belong to $F'$ which implies that $x$ and $y$ are non-adjacent in $G'[F']$.
This, however, contradicts the fact that $G[F']$ induces a graph isomorphic to $F$, because $x$ and $y$ are adjacent in $G$.
Thus $\mathcal{F}_{G'}\subset \mathcal{F}_{G}$ which implies that the described set of edges $E'$ constitutes a solution.
This actually corresponds to the \textsc{$d$-Hitting Set} problem: given a collection of sets $C_i=E_i$ each of size $d$ from a universe $U = E(G)$, select at most $\ell$ elements from $U$ such that every set $C_i$ contains a selected element. Then we use the result of Abu-Khzam~\cite{Abu-Khzam10} (see also~\cite{CyganFKLMPPS15}) that  \textsc{$d$-Hitting Set} admits a polynomial kernel with the universe size  $\Oh(\ell^d)$  and with $\Oh(\ell^d)$ sets.
\end{proof}

We would like to underline that Theorems~\ref{theo:weak:bounded:fpt} and \ref{theo:weak:bounded:polykernel} are fulfilled for the case when $F$ is a fixed graph of constant size,
as the degree of the polynomial in the running time of our algorithm for depends on the size of $F$ and, similarly, the size of $F$ is in the exponent of the function defining the size of our bi-kernel.
We can hardly avoid this dependence as it can be observed that for $\ell =0$,  \textsc{Strong $F$-closure} is equivalent to asking whether the input graph $G$ is $F$-free, that is,
we have to solve the \textsc{Induced Subgraph Isomorphism} problem. It is well known that \textsc{Induced Subgraph Isomorphism} parameterized by the size of $F$ is \W{1}-hard
when $F$ is a complete graph or graph without edges~\cite{DowneyF13}, and the problem is \W{1}-hard when $F$ belongs to other restricted families of graphs \cite{KhotR02}.

\medskip
We conclude with a few open problems.
An interesting question is whether \textsc{Strong Triadic Closure} is FPT when parameterized by 
$r= k-\mu(G)$. We proved that this holds on graphs of maximum degree at most 4, and we believe that this question is interesting not only on general graph but also on  various graph other classes. In particular, what can be said about planar graphs? To set the background, we show that \textsc{Strong Triadic Closure}
is NP-hard on this class. The following lemma is needed for the proof of Theorem \ref{theo:strong:planar}.

\begin{lemma}[\cite{KP17}]\label{lem:twins}
Let $x$ and $y$ be true twins in $G$ and let $H$ be a solution for \textsc{Strong $P_3$-closure}.
Then $xy \in E(H)$ and for every vertex $u \in N(x)$, $xu \in E(H)$ if and only if $yu \in E(H)$.
\end{lemma}

\begin{figure}[t]
\centering
\includegraphics[scale=1.0]{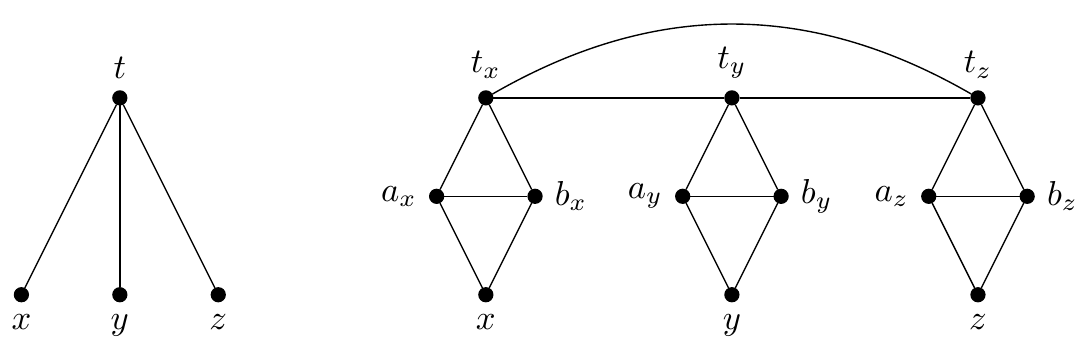}
%\vspace*{-0.1in}
\caption{The planar configuration used in the proof of Theorem~\ref{theo:strong:planar}.}\label{fig:figGadget}
%\vspace*{-0.1in}
\end{figure}

\begin{theorem}\label{theo:strong:planar}
\textsc{Strong Triadic Closure} is NP-hard on planar graphs.
\end{theorem}

\begin{proof}
We show the theorem by a reduction from {\sc PlanarX3C}.
In {\sc X3C} we are given a set $X$ with $|X|=3q$ elements and a collection $C$ of triplets of $X$ and the problem asks for a subcollection $C' \subseteq C$ such that every element of $X$ occurs in exactly one member of $C'$.
For the {\sc PlanarX3C} we associate a bipartite graph $G$ with this instance as follows: we have a vertex for every element of $X$ and
a vertex for every triplet of $C$ and there is an edge between an element and a triplet if and only if the element belongs to the triplet.
If the instance results in a graph $G$ that is (bipartite) planar then the problem is known to be NP-complete \cite{DF86}.
Let $G=(X \cup C,E)$ be an instance of {\sc PlanarX3C} with $|C|=m \geq q$.
We construct another graph $G^\prime$ by replacing the three edges incident to each triplet with the configuration shown in Figure~\ref{fig:figGadget}.
More precisely, we replace each triplet vertex $t$ by a triangle $\{t_x,t_y,t_z\}$ ({\it middle triangle}) and for each original edge $tx$ we introduce two triangles $\{t_x,a_x,b_x\}$ ({\it inner triangle}) and $\{a_x,b_x,x\}$ ({\it outer triangle}).
Thus for every triplet we associate seven triangles in which four of them are vertex-disjoint (the middle and the outer triangles) and the other three triangles (inner triangles) share all their vertices with two vertex-disjoint triangles.
Such a subgraph corresponding to the triplet $(x,y,z) \in C$ is simply called {\it triplet subgraph}.
Notice that $G^\prime$ remains a planar graph.
We prove that {\sc PlanarX3C} has an exact cover if and only if $G^\prime$ has a spanning subgraph with at least $9m+3q$ strong edges that satisfies the $P_3$-closure.

Assume $C'$ is an exact cover for {\sc PlanarX3C} with $|C'|=q$.
If a triplet belongs to $C'$ then we make the edges of all four vertex-disjoint triangles strong (see Figure~\ref{fig:figStrong}~(a)).
If a triplet does not belong to $C'$ then we make the edges of all inner triangles strong (see Figure~\ref{fig:figStrong}~(b)).
This labeling satisfies the $P_3$-closure as there is no $P_3$ spanned by strong edges and the total number of strong edges is $12q + 9(m-q)$ which gives the claimed bound.

For the opposite direction, assume that $G^\prime$ has a spanning subgraph $H$ with at least $9m+3q$ strong edges.
Consider the graph induced by the vertices $\{t_x,a_x,b_x,x\}$ that corresponds to an original edge between an element $x$ and a triplet $t$.
Since $a_x,b_x$ are true twins in $G$, by Lemma~\ref{lem:twins}, $E(H)$ contains the edge $a_xb_x$ and $a_xb_x$ are also true twins in $H$.
The latter implies that either one of the two triangles $\{x,a_x,b_x\}$, $\{t_x,a_x,b_x\}$ belongs to $H$, or no such triangle belongs to $H$.
The same observation carries along the vertices $a_y,b_y$ and $a_z,b_z$.
Thus for every triplet subgraph, $E(H)$ contains all its outer triangles, or all its inner triangles, or a combination of some inner and outer triangles.
These cases correspond to the three forms given in Figure~\ref{fig:figStrong}.
We show that there exists an optimal solution $H$ only with the first two forms of Figure~\ref{fig:figStrong},
which particularly means that every triplet subgraph of $H$ contains either all its outer triangles or all its inner triangles.

\begin{figure}[t]
\centering
\includegraphics[width=1.0\textwidth]{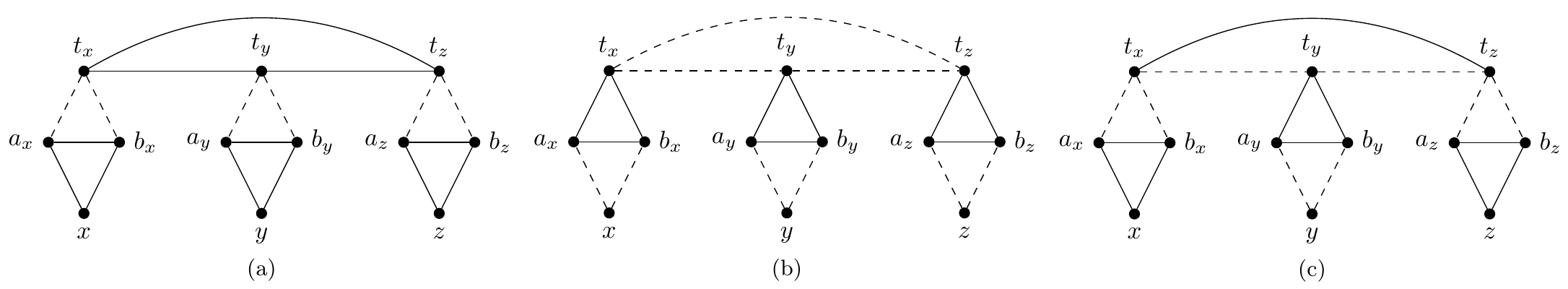}
%\vspace*{-0.1in}
\caption{A solid edge corresponds to a strong edge, whereas a dashed edge corresponds to a weak edge.
Form (a) has 12 strong edges and corresponds to a triplet that is a member of an exact cover.
Form (b) has 9 strong edges and corresponds to a triplet that does not belong to an exact cover.
Form (c) contains all other cases; we depict only one of them.}\label{fig:figStrong}
\end{figure}

To prove this, we first show that every middle triangle in a triplet subgraph has either all its edges strong or none of its edges is strong. We refer to the former case as strong middle triangle and the later as weak middle triangle.
Assume that the middle triangle contains at least one strong edge $t_xt_z$.
Then there is no other strong edge incident to $t_x$ or $t_z$.
If the inner triangle of $t_y$ is not strong, then we can safely make the edges $t_yt_x,t_yt_z$ strong.
Otherwise, the inner triangle of $t_y$ is strong and we remove both edges $t_ya_y$ and $t_yb_y$ from $H$ and add the edges $t_yt_x,t_yt_z$.
Thus if there is a strong edge in the middle triangle then there is a solution with a strong middle triangle.

Next we consider a (strong or weak) middle triangle.
If such a triangle is weak then $E(H)$ contains at most $9$ edges from its triplet subgraph.
In such a case we replace all its edges from $E(H)$ by the edges of its inner triangles by keeping the same size for $E(H)$.
For every strong middle triangle notice that all the edges of its inner triangles are weak.
If there is at most one outer triangle that is strong then we make the middle triangle weak and we replace its edges of $E(H)$ by the edges of its inner triangles.
Thus for every strong middle triangle we know that either two or three outer triangles are strong.
Also recall that for every weak middle triangle, all its outer triangles are weak.

For $i\in \{0,2,3\}$, let $\ell_i$ be the number of triplet subgraphs in which there are $i$ outer triangles strong.
We will show that, since $H$ contains at least $9m+3q$ edges, there are no triplet subgraphs with exactly two outer triangles strong, i.e., $\ell_2=0$.
Observe that $\ell_0 + \ell_2 + \ell_3 = m$.
Also notice that each of the subgraphs corresponding to $\ell_0$ contains 9 strong edges, $\ell_2$ contains 10 strong edges, and $\ell_3$ contains 12 strong edges.
Therefore the total number of strong edges is $9\ell_0 + 10\ell_2 + 12\ell_3$.
As $H$ contains at least $9m+3q$ edges, we get $\ell_2+3\ell_3 \geq 3q$.
Now notice that every vertex of $X$ is incident to at most one strong triangle.
Thus for each of the $\ell_2$ subgraphs there are $2$ vertices in $X$ that are incident to strong edges, whereas for each of the $\ell_3$ subgraphs there are $3$ such vertices in $X$. This implies that $2\ell_2 + 3\ell_3 \leq |X|=3q$.
Therefore $|E(H)| \geq 9m+3q$ holds only if $\ell_2=0$ and $\ell_3=q$, so that all triplet subgraphs with strong middle triangles correspond to an exact cover for the elements of $X$.
\end{proof}

The same question can be asked for the case when $F\neq P_3$ has a connected component with at least three vertices. As a first step, we give an FPT result when $F$ is a star.

\begin{theorem}\label{theo:strong:matching:k13:fpt}
For every $t\geq 3$, \textsc{Strong $K_{1,t}$-closure} can be solved in time $2^{\Oh(r^2)} \cdot n^{\mathcal{\Oh}(1)}$, where $r=k-\mu(G)$ .
\end{theorem}

\begin{proof}
We prove the theorem by constructing a kernel for the problem.

Let $G$ be a graph and $M$ be a maximum matching of $G$. Let $V_M$ be the set of vertices of $G$ that are covered by $M$. Let $X$ be a subset of vertices of $V(G)$ and $A$ be a subset of edges of $G[X]$, both initially set to be empty. We add elements to $X$ and $A$ by performing the following steps in a greedy way:

\begin{enumerate}
\item If there is $v\in V(G)\setminus V_M$ and $xy\in M$ such that $vx\in E(G)$ or $vy\in E(G)$, then we add $v$, $x$ and $y$ to $X$ and add all the edges between $\{v,x,y\}$ to $A$.

\item If there is $xy,wz\in M$ such that $G[\{x,y,w,z\}]\not\simeq 2K_2$, then we add $x$, $y$, $w$ and $z$ to $X$ and add all the edges between $\{x,y,w,z\}$ to $A$.
\end{enumerate}

Note that since the set $\{v,x,y\}$ does not induce a $K_{1,t}$, and since $xy,wz\in E(G)$, the set $\{x,y,w,z\}$ does not induce a $K_{1,t}$ either, the edges added to $A$ in each step can be part of a solution. Moreover, after each application of step 1 or step 2, the size of the set $A\cup M$ is increased by at least one unity. As a consequence, if the steps can be applied at least $r$ times, then $|A\cap M|\geq |M|+r$ and therefore we have a yes instance.

After the exhaustive application of steps 1 and 2 in a greedy way, we consider a maximum matching in $G-X$. For simplicity, we call this matching $M$ again. Since step 2 can no longer be applied, $M$ is actually an induced matching of $G$. Moreover, since after every application of any of the steps at most 4 vertices were added to $X$ and since the steps have been applied at most $r-1$ times, we have that $|X|<4r$.

In what follows, we show that, in a yes instance, the size of $M$ can also be bounded by a function of $r$. We can partition the edges of $M$ according to their neighborhood inside $X$. There are at most $2^{4r}$ possible subsets of $X$ that can be the neighborhood of a given vertex of $G-X$. Then, we can partition the edges of $M$ into at most $2^{8r}$ classes, according to the neighborhoods of the two endpoints of the edge. We exhaustively apply the following rule.

\begin{Rule} \label{rule:k13:03}
If there exists a class of edges of $M$ that has size at least $(2t-2)\cdot 4r+1$, then delete one edge of the class from the graph and decrease the parameter by one.
\end{Rule}

To see that rule~\ref{rule:k13:03} is safe, assume that one given class contains at least $(2t-2)\cdot 4r+1$ edges. Note that, since $M$ is an induced matching in $G$, every vertex of $X$ can be adjacent to at most $2t-2$ vertices of $G-X$ in a solution, otherwise the  solution would contain a set of strong edges inducing a $K_{1,t}$ in $G$. This, together with the fact that $|X|<4r$ gives us that at most $(2t-2)\cdot 4r$ vertices of $G-X$ are adjacent to vertices of $X$ in a solution. Since the class contains at least $(2t-2)\cdot 4r+1$ edges, at least one edge of the class is such that both of its end points are not adjacent to any vertex of $X$ in the solution. This edge can therefore be part of every solution $H$.

Once rule~\ref{rule:k13:03} has been exhaustively applied, the number of vertices of the graph is bounded by $4r+2^{8r}\cdot (2t-2)\cdot 4r\cdot 2=g(r)$. It is now possible to use brute force to solve the problem in the following way. First, we guess which edges inside $X$ go into the solution. Since $|X|<4r$, this guessing takes $2^{O(r^2)}$ time. Since every vertex of $X$ can have at most $2t-2$ neighbors in $G-X$ in a solution, we can again guess which edges from $X$ to $G-X$ go into the solution. This takes $2^{O(r^2)}$ time. Finally, for each of these guesses made for the edges in $E(G)\setminus M$, we test which edges of $M$ can be added into the solution without forming an induced $K_{1,t}$ in $H$ that also induce a $K_{1,t}$ in $G$. This takes time $2^{\Oh(r)}$. The total running time of the brute force algorithm is therefore $2^{\Oh(r^2)}\cdot n^{\Oh(1)}$.
\end{proof}

Another direction of research is to extend \textsc{Strong $F$-closure} by replacing $F$ with a list of forbidden subgraphs $\mathcal{F}$ %.
and settle the complexity differences compared to $\mathcal{F}=\{F\}$.

%\bibliographystyle{siam}
%\bibliography{closure}

\end{document}